%% file: mfps.tex
\documentclass[9pt,a4paper]{entcs}
\usepackage{entcsmacro}
\usepackage{graphicx}
\usepackage{subfig}
\usepackage{wrapfig}

\usepackage{combelow}
\usepackage{cite}
\usepackage{tikzit}
\input{stringdiagrams.tikzstyles}

\usetikzlibrary{positioning,arrows,decorations.pathmorphing}
\usepackage{tikz-cd}

\usepackage{multicol}
\usepackage{enumerate}
\usepackage{amssymb,amsmath}
\usepackage{mathrsfs}
\usepackage{mathtools}

\newcommand{\tot}{\xrightarrow}
\newcommand{\dec}[1]{\ensuremath{\langle #1 \rangle}}
\newcommand{\ridm}[1]{\ensuremath{\overline{#1}}}
\newcommand{\iso}{\cong}
\newcommand{\entails}{\vDash}
\newcommand{\dle}{\sqsubseteq}
\newcommand{\codiag}{\nabla}
\newcommand{\fromm}{\xhookleftarrow}

\newcommand{\KT}{\ensuremath{\mathbf{K}_{\mathrm{3}}}}
\newcommand{\KTw}{\ensuremath{\mathbf{K}^{\mathrm{w}}_{\mathrm{3}}}}
\newcommand{\BT}{\ensuremath{\mathbf{B}_{\mathrm{3}}}}

\newcommand{\C}{\ensuremath{\mathscr{C}}}

\newcommand{\Pfn}{\ensuremath{\mathbf{Pfn}}}
\newcommand{\DMQLat}{\ensuremath{\mathbf{DMQLat}}}
\newcommand{\Poset}{\ensuremath{\mathbf{Poset}}}
\newcommand{\Set}{\ensuremath{\mathbf{Set}}}

\DeclareMathOperator{\id}{id}
\DeclareMathOperator{\Dec}{Dec}
\DeclareMathOperator{\TDec}{TDec}
\DeclareMathOperator{\Hom}{Hom}
\DeclareMathOperator{\Inv}{Inv}
\DeclareMathOperator{\Total}{Total}

\DeclareMathOperator{\Par}{Par}
\DeclareMathOperator{\Split}{Split}

\newcommand{\ie}{\emph{i.e.}}
\newcommand{\eg}{\emph{e.g.}}

\newif\ifappendix
\appendixtrue

\def\lastname{Kaarsgaard}
\begin{document}

\begin{frontmatter}
  \title{Condition/Decision Duality and the \\ Internal Logic of Extensive 
  Restriction Categories}
  \author{Robin Kaarsgaard\thanksref{myemail}\thanksref{ALL}
  }
  \address{DIKU, Department of Computer Science\\ University of Copenhagen}
  \thanks[myemail]{Email: \href{mailto:robin@di.ku.dk}
  {\texttt{\normalshape robin@di.ku.dk}}}
  \thanks[ALL]{The author would like to thank Robert Glück for discussions
  relating to this paper, and to acknowledge the support given by 
  \emph{COST Action IC1405 Reversible computation: Extending horizons of
  computing}. The string diagrams and flowcharts in this paper were produced
  using \href{https://tikzit.github.io/}{\emph{TikZiT}}.}

  \begin{abstract}
  In flowchart languages, predicates play an interesting double role. In the
  textual representation, they are often presented as \emph{conditions}, \ie, 
  expressions which are easily combined with other conditions
  (often via Boolean combinators) to form new conditions, though they only play 
  a supporting role in aiding branching statements choose a branch to follow.
  On the other hand, in the graphical representation they are typically
  presented as \emph{decisions}, intrinsically capable of directing control
  flow yet mostly oblivious to Boolean combination.
  
  While categorical treatments of flowchart languages are abundant, none of 
  them provide a treatment of this dual nature of predicates.
  In the present paper, we argue that extensive restriction categories are 
  precisely categories that capture such a condition/decision duality, by 
  means of morphisms which, coincidentally, are also called decisions.
  Further, we show that having these categorical decisions amounts to having an
  internal logic: Analogous to how subobjects of an object in a topos form a
  Heyting algebra, we show that decisions on an object in an extensive
  restriction category form a \emph{De Morgan quasilattice}, the algebraic
  structure associated with the (three-valued) \emph{weak Kleene logic} \KTw. 
  Full classical propositional logic can be recovered by restricting to
  \emph{total} decisions, yielding extensive categories in the usual sense, and
  confirming (from a different direction) a result from effectus theory that
  predicates on objects in extensive categories form Boolean algebras.

  As an application, since (categorical) decisions are partial isomorphisms,
  this approach provides naturally reversible models of classical propositional 
  logic and weak Kleene logic.
  \end{abstract}
\end{frontmatter}

\section{Introduction}
Flowchart languages are a particular class of imperative programming languages
which permit a pleasant and intuitive graphical representation of the control
flow of programs. While conceptually very simple, flowchart languages form the
foundation for modern imperative programming languages, and have been used for
this reason as vehicles for program analysis (\eg, to measure coverage in
white-box testing~\cite{Ammann2008}), program transformations (\eg, partial
evaluation, see \cite{JonesGomardSestoft1993}), and to express fundamental
properties of imperative programming, such as the equivalence of expressivity
in \emph{structured} and \emph{unstructured} programming in the
\emph{Böhm-Jacopini theorem}~\cite{BohmJacopini1966} (see also
\cite{AshcroftManna1972,WilliamsOssher1978}). Figure
\ref{fig:structured_flowchart} shows the (textual and graphical) flowchart
structures used by structured flowchart languages.

An interesting feature in flowchart languages is the dual presentation of
predicates as \emph{conditions} and \emph{decisions}, depending on the context.
On the one hand, the textual
$\mathbf{if}~p~\mathbf{then}~c_1~\mathbf{else}~c_2$ seems to favor the view of
$p$ as a condition, \ie, a predicate which has inherently nothing to do with
control flow, but which may easily be combined with other conditions to other
conditions to form new ones. In other words, the textual representation
considers the branching behaviour to be given by the semantics of
$\mathbf{if}~\dots~\mathbf{then}~\dots~\mathbf{else}~\dots$ rather than by the
semantics of $p$. This view is also emphasized by the usual (big-step)
operational semantics of conditionals: Here, predicates are treated as
expressions that may be evaluated in a state to yield a Boolean value, which
the conditional may then branch on, as in
\begin{equation*}
  \frac{\langle p, \sigma \rangle \to \mathbf{true} \qquad \langle
  c_1, \sigma \rangle \to \sigma'}{\langle
  \mathbf{if}~p~\mathbf{then}~c_1~\mathbf{else}~c_2, \sigma \rangle \to \sigma'}
  \qquad \text{ and } \qquad
  \frac{\langle p, \sigma \rangle \to \mathbf{false} \qquad \langle
  c_2, \sigma \rangle \to \sigma'}{\langle
  \mathbf{if}~p~\mathbf{then}~c_1~\mathbf{else}~c_2, \sigma \rangle \to \sigma'}
  \enspace.
\end{equation*}
On the other hand, the graphical representation of conditionals in
Figure~\ref{fig:structured_flowchart}(c) seems to rather prefer the view of
$p$ as a decision, \ie, a kind of flowchart operation intrinsically capable of
directing control flow. That is to say, that this is a structured flowchart
(corresponding to a conditional) is purely coincidental; for
\emph{unstructured} flowcharts to make sense, $p$ \emph{must} be able to direct
control flow on its own. However, where conditions are most naturally composed
via the Boolean combinators, the only natural way of composing decisions seems
to be in sequence (though this leads to additional output branches).

While categorical models of structured flowchart languages have been widely
studied (see, \eg, \cite{Stefanescu1986, Stefanescu1987, ManesArbib1980,
ManesArbib1986,Elgot1975,Elgot1976}), none provide a treatment of this dual
view of predicates. In this paper, we argue that extensive restriction
categories are precisely categories that make clear this dual view on
predicates as conditions and decisions, offering both the ease of combination
of conditions and the control flow behaviour of decisions. Restriction
categories (introduced in \cite{Cockett2002, Cockett2003, Cockett2007}) are
categories of partial maps, in which each morphism is equipped with a
\emph{restriction idempotent} that, in a certain sense, gauges how partial that
morphism is. Since models of flowchart languages most provide a notion of
partiality (due to possible nontermination), restriction categories provide an
ideal setting for such models. Coincidentally, the defining feature of
extensive restriction categories\footnote{Note that while extensive restriction
categories are strongly connected to extensive categories, they are confusingly
\emph{not} extensive in the usual sense of extensive
categories~\cite{Carboni1993}.} is the presence of certain morphisms called
\emph{decisions}, which play a similar role as the decision view on predicates
in flowchart languages.

In this setting, we show that the correspondence between conditions and
decisions is exhibited precisely as a natural isomorphism between the
\emph{predicate fibration} $\Hom(X,1+1)$ of predicates and predicate
transformers (see also \cite{CJWW2015,Jacobs2015}), and the \emph{decision
fibration} $\Dec(X)$ of decisions (certain morphisms $X \to X+X$) and decision
transformers. We then go on to explore the structure of $\Dec(X)$ (or
equivalently, $\Hom(X,1+1)$), showing that this extends to a fibration over the
category of \emph{De Morgan quasilattices} and homomorphisms, which give
algebraic semantics~\cite{FinnGrigolia1993} to Kleene's \emph{weak logic}
\KTw~\cite{Kleene1952}. Intuitively, \KTw{} can be seen as a partial version of
classical (Boolean) logic. We make this statement precise in this setting by
showing that if we restrict ourselves to \emph{total} decisions and decision
transformations, classical logic can be recovered. Since the subcategory of
objects and total morphisms of a (split) extensive restriction category is an
extensive category in the usual sense (see, \eg, \cite{Carboni1993}), we can
use this to provide an alternative proof of a statement from Effectus
theory~\cite{Jacobs2015,CJWW2015} that predicates over each extensive category forms a fibred Boolean algebra via the predicate fibration~\cite[Prop.~61,
Prop.~88]{CJWW2015}. This yields a relationship diagram of effecti, extensive
categories, and extensive restriction categories and their corresponding logics
as shown in Figure~\ref{fig:exteff}.

\begin{figure}
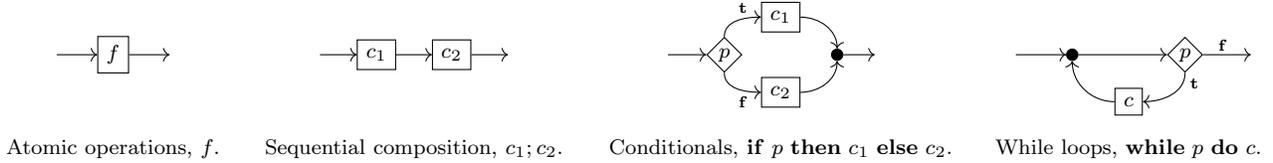

\ctikzfig{fc_struct}
\caption{The four flowchart structures.}
\label{fig:structured_flowchart}
\end{figure}

This paper is structured as follows: Section \ref{sec:background} gives a brief
introduction to extensive restriction categories. 
Section \ref{sec:condition_decision_duality} demonstrates the
condition/decision duality of extensive restriction categories by showing that
the decision and predicate fibrations are naturally isomorphic; and, as a
consequence, that \emph{decisions are a property of the predicates}. Then, in
Section~\ref{sec:the_internal_logic_of_extensive_restriction_categories}, we
show that the decisions on an object form models of \KTw, with decision
transformers as homomorphisms. By restricting to only \emph{total} decisions,
we show that these restrict to models of classical logic. Finally,
\ref{sec:conclusion} offers some concluding remarks.

\begin{figure}
  \centering
  \begin{tikzpicture}
  \node[draw,rounded corners] at (0.75,0) (ExtCat) {\begin{tabular}{c}
    \textbf{Extensive category} \\
    \emph{Boolean algebras}
  \end{tabular}};

  \node[draw,rounded corners] at (5.25,-4) (ExtRestCat) {\begin{tabular}{c}
    \textbf{Extensive restriction category} \\
    \emph{De Morgan quasilattices}
  \end{tabular}};

  \node[draw,rounded corners] at (-5.5,-4) (Effectus) {\begin{tabular}{c}
    \textbf{Effectus} \\
    \emph{Effect algebras}
  \end{tabular}};

  \draw[->] (Effectus) to node {} (ExtCat);
  \draw[->] (ExtRestCat) to node {} (ExtCat);
  \end{tikzpicture}
  \caption{Extensive categories, extensive restriction categories, and effecti: 
  their relationships and associated logics.}
  \label{fig:exteff}
\end{figure}
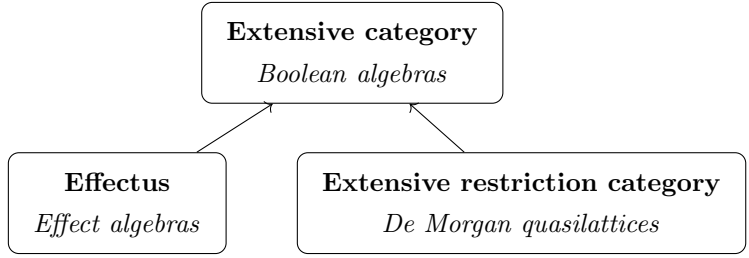

\section{Extensive restriction categories} 
\label{sec:background}
This section gives an introduction to extensive
restriction categories 
as it will be applied in the
sections that follow. The experienced reader may safely skip this section on a
first reading, and instead refer back to it as necessary.

Restriction categories are categories equipped with notions of partiality and
totality of morphisms. This is done by means of a \emph{restriction
combinator}, assigning to each morphism $X \tot{f} Y$ its \emph{restriction
idempotent} $X \tot{\ridm{f}} X$ (subject to certain laws) which may
intuitively be thought of as a partial identity defined precisely where $f$ is
defined. In this way, restriction categories provide an axiomatic (and
relatively light-weight) approach to partiality of morphisms in categories.
Formally, restriction categories are defined in the following way:

\begin{definition}
  A \emph{restriction structure} on a category consists of a combinator mapping 
  each morphism $f$ to its \emph{restriction idempotent} $\ridm{f}$, \ie
  \begin{equation*}
    \frac{X \tot{f} Y}{X \tot{\ridm{f}} X}
  \end{equation*}
  subject to the \emph{restriction laws}:
  \begin{multicols}{2}
  \begin{itemize}
    \setlength{\itemindent}{3em}
    \item[(R1)] $f \ridm{f} = f$ for all $X \tot{f} Y$,
    \item[(R2)] $\ridm{f} \ridm{g} = \ridm{g} \ridm{f}$ for all $X \tot{f} Y$
    and $X \tot{g} Z$,
    \item[(R3)] $\ridm{f\ridm{g}} = \ridm{f} \ridm{g}$ for all $X \tot{f} Y$
    and $X \tot{g} Z$, and
    \item[(R4)] $\ridm{g} f = f \ridm{gf}$ for all $X \tot{f} Y$ and $Y \tot{g}
    Z$.
  \end{itemize}
  \end{multicols}
  A category equipped with a restriction structure is called a 
  \emph{restriction category}.
\end{definition}

As the name suggests, a restriction structure is a structure on a category
rather than a property of it; in particular, a category can be equipped with
several different restriction structures. For this reason, we must in principle
specify which restriction structure we are using when speaking of a particular
category as a restriction category, though this is often omitted when the
restriction structure is implicitly given to be a canonical one.

Given that restriction categories are built on a foundation of idempotents, one would expect it to be occasionally useful when all such restriction idempotents split, and indeed this is the case. Say that restriction structure is \emph{split} when all restriction idempotents split, and let $\Split(\C)$ denote the category arising from the usual idempotent splitting (\ie, the \emph{Karoubi envelope}) of all restriction idempotents in \C. That $\Split(\C)$ is a restriction category when $\C$ is follows by \cite[Prop.~2.26]{Cockett2002}

As a canonical example,
the category \Pfn{} of sets and partial functions is a restriction
category, with the restriction idempotent $X \tot{\ridm{f}} X$ for $X \tot{f}
Y$ given by
\begin{equation*}
  \ridm{f}(x) = \left\{ \begin{array}{ll}
    x & \text{if $f$ is defined at $x$} \\
    \text{undefined}\enspace & \text{otherwise}
  \end{array}\right.
\end{equation*}
In a restriction category, say that a morphism $X \tot{f} Y$ is \emph{total} if
$\ridm{f} = \id_X$, and that it is a \emph{partial isomorphism} if there exists
$Y \tot{f^\dagger} X$ such that $f^\dagger f = \ridm{f}$ and $ff^\dagger =
\ridm{f^\dagger}$. Partial isomorphisms thus generalize ordinary isomorphisms,
as an isomorphism is then a partial isomorphism $X \tot{f} Y$ such that both
$f$ and $f^\dagger$ are total.

Since total morphisms are closed under composition and include all identities,
they form an important subcategory $\Total(\C)$ of any restriction category \C.
Likewise, partial isomorphisms are closed under composition and include all
identities, so all objects and partial isomorphisms of \C{} form the
subcategory $\Inv(\C)$. As the notation suggests, this category $\Inv(\C)$ is
not just a restriction category but an \emph{inverse category} (indeed, it is
the \emph{cofree} such~\cite{Kaarsgaard2017}) in the usual sense (see
\cite{Cockett2002,Kastl1979}).

A useful property of restriction categories is that they come with a natural
partial order on homsets (which extends to enrichment in \Poset) given by $f
\le g$ iff $g \ridm{f} = f$. Intuitively, this can be thought of as an
\emph{information order}; $f \le g$ if $g$ can do everything $f$ can do, and
possibly more.

Like any other categorical structure, when working in restriction categories we
require everything in sight to cooperate with the restriction structure. One of
the simplest examples of cooperation with restriction structure is given in the
definition of a restriction terminal object: This is simply a terminal object
$1$ in the usual sense, which further satisfies that the unique map $X \to
1$ is \emph{total} for all objects $X$. For coproducts, this means that we not
only require the restriction category to have all finite coproducts in the
usual sense, but also that the coproduct injections $X \tot{\kappa_1} X+Y$ and
$Y \tot{\kappa_2} X+Y$ are \emph{total}. In this case, we say that the
restriction category has \emph{restriction coproducts}. There is also a similar
notion of a \emph{restriction zero} object $0$: a zero object in the usual
sense which additionally satisfies that each zero endomorphism $X \tot{0_{X,X}}
X$ is its own restriction idempotent, \ie, that $\ridm{0_{X,X}} = 0_{X,X}$ (or
equivalently, that $\ridm{0_{X,Y}} = 0_{X,X}$ for all zero morphisms
$0_{X,Y}$). When zero morphisms exist, they serve as least element in their
homset with respect to the natural ordering, and when a category has
restriction coproducts and a restriction zero object, the restriction zero
object serves as unit for the restriction coproduct. When this is the case,
restriction coproduct injections are further partial isomorphisms (\eg, the partial inverse to $X \tot{\kappa_1} X+Y$ is $X+Y \tot{[\id,0]} X$).

Extensivity for restriction categories means that the restriction coproducts
are particularly well-behaved, in the sense that they admit a \emph{calculus of
matrices}~\cite{Cockett2007}. Concretely, this means that each morphism $X
\tot{f} Y+Z$ is associated with a unique morphism $X \tot{\dec{f}} X+X$, its
\emph{decision}, which, intuitively, makes the same branching choices as $f$
does, but doesn't do any actual work. Extensive restriction categories are
defined as follows.

\begin{definition}\label{def:extensive}
  A restriction category with restriction coproducts and a restriction zero is
  said to be an \emph{extensive restriction category} if each morphism $f$ has
  a unique \emph{decision} $\dec{f}$, \ie
  \begin{equation*}
    \frac{X \tot{f} Y+Z}{X \tot{\dec{f}} X+X}
  \end{equation*}
  satisfying the \emph{decision laws}
  \begin{multicols}{2}
  \begin{itemize}
    \setlength{\itemindent}{3em}
    \item[(D1)] $\codiag \dec{f} = \ridm{f}$
    \item[(D2)] $(f+f) \dec{f} = (\kappa_1 + \kappa_2) f$
  \end{itemize}
  \end{multicols}
  where $X+X \tot{\codiag} X$ is the natural codiagonal $[\id,\id]$.
\end{definition}
Note that extensive restriction categories are not extensive in the usual sense
-- rather, extensive restriction categories are the ``partial'' version of
extensive categories. This connection is made precise by the following proposition due to \cite{Cockett2007}.
\begin{proposition}
  Whenever \C{} is an extensive restriction category, $\Total(\Split(\C))$ is
  an extensive category.
\end{proposition}
A straightforward example of an extensive restriction category is \Pfn. Here,
the decision $X \tot{\dec{f}} X+X$ of a partial function $X \tot{f} Y+Z$ is
given by
\begin{equation*}
  \dec{f}(x) = \left\{ \begin{array}{ll}
    \kappa_1(x) & \text{if $f(x) = \kappa_1(y)$ for some $y \in Y$} \\
    \kappa_2(x) & \text{if $f(x) = \kappa_2(z)$ for some $z \in Z$} \\
    \text{undefined} & \text{if $f$ undefined at $x$}
  \end{array}\right.
\end{equation*}
For further examples and details on extensive restriction categories, see \cite{Cockett2007}.

\section{Condition/decision duality} 
\label{sec:condition_decision_duality}
Categorical models of flowcharts are categories with a notion of partiality
(due to possible nontermination) and coproducts (corresponding to the control
flows of the flowchart). As such, restriction categories with restriction
coproducts serve as a good starting point for these. We show in this section
that the additional requirement of \emph{extensivity} of the restriction
coproduct allows the category to exhibit a condition/decision duality,
analogous to the flowchart languages. This manifests in the category as a
natural isomorphism between the \emph{decisions} and \emph{predicates} over an
object (with their corresponding transformations).

We start with a few technical lemmas regarding the partial order on morphisms
in a restriction category as well as properties of decisions in extensive
restriction categories.

\begin{lemma}\label{lem:ridm}
  It is the case that
  \begin{multicols}{2}
  \begin{enumerate}[(i)]
    \item\label{lem:ridm:1} $g \le g'$ implies $\ridm{gf} \le \ridm{g'f}$,
    \item\label{lem:ridm:2} $\ridm{gf} \le \ridm{f}$,
    \item\label{lem:ridm:3} $f \le g$ implies $hf \le hg$ and $fh' \le gh'$,
    \item\label{lem:ridm:4} $f \le f'$ and $g \le g'$ iff $f+g \le f'+g'$.
    \item\label{lem:ridm:5} $f \le \ridm{g}$ implies $f = \ridm{f}$.
  \end{enumerate}
  \end{multicols}
\end{lemma}

\begin{lemma}\label{lem:utility}
  Let $X \tot{f} Y+Z$ and $X' \tot{g} X$ be arbitrary morphisms of an extensive 
  restriction category, and $X \tot{\ridm{e}} X$ any restriction idempotent. It 
  is the case that
  \begin{multicols}{2}
    \begin{enumerate}[(i)]
    \item\label{lem:1} $\dec{\mkern-3mu\dec{f}\mkern-3mu} = \dec{f}$
    \item\label{lem:2} $\dec{f}$ is a partial isomorphism and $\dec{f}^\dagger 
    =  \left[\,\ridm{\kappa_1^\dagger f},\ridm{\kappa_2^\dagger f}\,\right]$
    \item\label{lem:2.5} $\ridm{\dec{f}^\dagger} = \ridm{\kappa_1^\dagger f} + 
    \ridm{\kappa_2^\dagger f}$
    \item\label{lem:3} $\ridm{\dec{f}} = \ridm{f}$
    \item\label{lem:4} $\gamma \dec{f} = \dec{\gamma f}$
    \item\label{lem:5} $\dec{\mkern-3mu\dec{f}g} = \dec{fg}$
    \item\label{lem:6lem} $(\ridm{e}+\ridm{e})\dec{f} = 
    (\ridm{e}+\ridm{e})\dec{f}\ridm{e}$
    \item\label{lem:6} $\dec{f}\ridm{e}$ is a decision and $\dec{f}\ridm{e} =
    (\ridm{e}+\ridm{e})\dec{f}$
    \item\label{lem:7} $\dec{f}\ridm{e} = \dec{f \ridm{e}}$
    \item\label{lem:8} $\kappa_i^\dagger \dec{f} = \ridm{\kappa_i^\dagger f}$
    \item\label{lem:9} $\dec{g}f = (f+f)\dec{gf}$
    \end{enumerate}
  \end{multicols}
\end{lemma}

A few of these identities were shown already in \cite{Cockett2007}; the rest
are mostly straightforward to derive. Note that a direct consequence of
\eqref{lem:1} is that $(\dec{f}+\dec{f})\dec{f} = (\kappa_1 +
\kappa_2)\dec{f}$; we will make heavy use of this fact in
Section~\ref{sec:the_internal_logic_of_extensive_restriction_categories}.
Another particularly useful identity is the following, stating intuitively that
anything that behaves as a decision in each component is, in fact, a decision.

\begin{lemma}\label{lem:decrep}
  If $\kappa_1^\dagger p = \ridm{\kappa_1^\dagger f}$ and $\kappa_2^\dagger p = 
  \ridm{\kappa_2^\dagger f}$ then $p = \dec{f}$.
\end{lemma}
\begin{proof}
  By Lemma~\ref{lem:utility}\eqref{lem:2} $\dec{f}^\dagger =
  \left[\,\ridm{\kappa_1^\dagger f},\ridm{\kappa_2^\dagger f}\,\right]$. Since
  $\kappa_1^\dagger p = \ridm{\kappa_1^\dagger f}$ and
  $\kappa_2^\dagger p = \ridm{\kappa_2^\dagger f}$ it follows that
  $\ridm{\kappa_1^\dagger f} = (\ridm{\kappa_1^\dagger f})^\dagger =
  (\kappa_1^\dagger p)^\dagger = p^\dagger \kappa_1$ and
  $\ridm{\kappa_2^\dagger f} = (\ridm{\kappa_2^\dagger f})^\dagger =
  (\kappa_2^\dagger p)^\dagger = p^\dagger \kappa_2$
  so it follows by the universal mapping property of the coproduct that
  $p^\dagger = \left[\,\ridm{\kappa_1^\dagger f},\ridm{\kappa_2^\dagger
  f}\,\right] = \dec{f}^\dagger$, and finally $p = \dec{f}$ by unicity of
  partial inverses.
\end{proof}

As a corollary, $p$ is a decision if $\kappa_1^\dagger p$ and $\kappa_2^\dagger
p$ are both restriction idempotents (\ie, if $\kappa_1^\dagger p =
\ridm{\kappa_1^\dagger p}$ and $\kappa_2^\dagger p = \ridm{\kappa_2^\dagger
p}$) since all decisions decide themselves (\ie, since $\dec{\mkern-3mu\dec{p}\mkern-3mu} = \dec{p}$).

\begin{theorem}
  There is a functor $\C^{\text{op}} \tot{\Dec} \Set$ given by mapping objects
  to their decisions, and morphisms to decision transformers.
\end{theorem}
\begin{proof}
  Define this functor by $\Dec(X) = \{\dec{p} \mid p
  \in \Hom(X,X+X)\}$ on objects, and by $\Dec(f : Y \to X)(\dec{p}) =
  \dec{\mkern-3mu\dec{p}f}$ on morphisms. This is contravariantly functorial 
  since $\Dec(\id_X)(\dec{p}) = \dec{\mkern-3mu\dec{p}\id_X} =
  \dec{\mkern-3mu\dec{p}\mkern-3mu} = \dec{p}$ by
  Lemma~\ref{lem:utility}\eqref{lem:1}, and $\Dec(gf)(\dec{p}) =
  \dec{\mkern-3mu\dec{p} gf} = \dec{\mkern-3mu\dec{\mkern-3mu\dec{p} g}f} =
  \Dec(f)(\Dec(g)(\dec{p}))$ by Lemma~\ref{lem:utility}\eqref{lem:5} and
  definition of $\Dec(f)$, as desired.
\end{proof}

From now on, we will use the notation $\Dec(Y) \tot{f^\diamond} \Dec(X)$ for 
the decision transformation $\Dec(f)$.

This is an example of a fibred category, which have historically been important
in categorical presentations of logic, \eg, in topoi (see \cite{Jacobs1999} for
a thorough treatment of indexed and fibred categories in categorical logic). In
Section~\ref{sec:the_internal_logic_of_extensive_restriction_categories}, we
will see that this indexed category extends beyond $\Set$ to a model of \KTw.
For now, it is sufficient to show the equivalence between conditions (morphisms
$X \to 1+1$) and decisions (morphisms $X \to X+X$ satisfying the \emph{decision
laws} of Definition~\ref{def:extensive}).

\begin{theorem}[Condition/decision duality]\label{thm:dec_cond_duality}
  Decisions and predicates are naturally isomorphic in any extensive 
  restriction category with a restriction terminal object: $\Dec(-) \iso 
  \Hom(-,1+1)\enspace$.
\end{theorem}
\begin{proof}
  Let \C{} be an extensive restriction category with a restriction terminal
  object, and $X$ some object of \C; we begin by showing that the mappings
  \begin{equation*}
    X \tot{\dec{f}} X+X\quad \mapsto\quad X\tot{\dec{f}} X+X \tot{!+!} 1+1 
    \qquad\text{and}\qquad
    X \tot{p} 1+1 \quad \mapsto \quad X \tot{\dec{p}} X+X 
  \end{equation*}
  between $\Dec(X)$ and $\Hom(X,1+1)$ yields a bijection. In other words, we 
  must show that $\dec{\mkern-2mu(!+!)\dec{f}\mkern-3mu} = \dec{f}$ and 
  $p = (!+!)\dec{p}$. To show $\dec{\mkern-2mu(!+!)\dec{f}\mkern-3mu} = 
  \dec{f}$, we show that $\dec{f}$ decides $(!+!)\dec{f}$ by
  $\nabla \dec{f} = \ridm{f} = \ridm{\dec{f}} = \ridm{(!+!)\dec{f}}$
  using the fact that the unique map $X \tot{!} 1$ is total by $1$ restriction 
  terminal, and by
    $(((!+!)\dec{f}) + ((!+!)\dec{f})) \dec{f} = ((!+!)+(!+!)) (\dec{f} + 
    \dec{f}) \dec{f}
    = ((!+!)+(!+!)) (\dec{f} + \dec{f}) \dec{\mkern-3mu\dec{f}\mkern-3mu}
    = ((!+!)+(!+!)) (\kappa_1 + \kappa_2) \dec{f}
    = (\kappa_1 + \kappa_2) (!+!) \dec{f}$.
  Thus $\dec{\mkern-2mu(!+!)\dec{f}\mkern-3mu} = \dec{f}$, as desired.
  
  To show that $p = (!+!)\dec{p}$ for $X \tot{p} 1+1$ we show something
  slightly more general, namely that $(!+!)f = (!+!)\dec{f}$ for any $X \tot{f}
  Y+Z$. That $p = (!+!)\dec{p}$ then follows as a special case since $\id_{1+1}
  = (!+!)$ by $1$ terminal, so $p = \id_{1+1} p = (!+!) p$. This slightly more 
  general statement follows by commutativity of the diagram below.
  \begin{center}
  \begin{tikzpicture}
  \node (X) {$X$};
  \node[right=25mm of X] (XX) {$X+X$};
  \node[below=12mm of X] (11') {$Y+Z$};
  \node[below=12mm of XX] (1111) {$(Y+Z)+(Y+Z)$};
  \node[below right=10mm of 1111] (11) {$1+1$};
  
  \node[below=5mm of X] (p1) {};
  \node[right=11mm of p1] (i) {\footnotesize\emph{(i)}};
  \node[below=5mm of 1111] (ii) {\footnotesize\emph{(ii)}};
  \node[above=2mm of 1111] (phantom) {};
  \node[right=8mm of phantom] (iii) {\footnotesize\emph{(iii)}};
  
  \draw[->,font={\small}] (X) to node [above] {$\dec{f}$} (XX);
  \draw[->,font={\small}, bend left] (XX) to node [right] {$!+!$} (11);
  \draw[->,font={\small}] (X) to node [left] {$f$} (11');
  \draw[->,font={\small}, bend right] (11') to node [below] {$!+!$} (11);
  \draw[->,font={\small}] (11') to node [below] {$\kappa_1 + \kappa_2$} 
  (1111);
  \draw[->,font={\small}] (XX) to node [right] {$f+f$} (1111);
  \draw[->,font={\small}] (1111) to node [above right] {$!+!$} (11);
  \end{tikzpicture}
  \end{center}
  Here, \emph{(i)} commutes by the second axiom of decisions, while 
  \emph{(ii)} and \emph{(iii)} both commute by $1$ terminal.
  
  To see that this bijection extends to a natural isomorphism, we must fix some
  $Y \tot{f} X$ and chase the diagram
  \begin{center}
  \begin{tikzpicture}
  \node (DecX) {$\Dec(X)$};
  \node[right=15mm of DecX] (HomX) {$\Hom(X,1+1)$};
  \node[below=10mm of DecX] (DecY) {$\Dec(Y)$};
  \node[below=10mm of HomX] (HomY) {$\Hom(Y,1+1)$};
  \draw[->,font={\small}] (DecX) to node [left] {$f^\diamond$} (DecY);
  \draw[<->,font={\small}] (DecX) to node [above] {$\iso$} (HomX);
  \draw[<->,font={\small}] (DecY) to node [below] {$\iso$} (HomY);
  \draw[->,font={\small}] (HomX) to node [right] {$f^*$} (HomY);
  \end{tikzpicture}
  \end{center}
  where we use $f^\diamond$ to denote the functorial action $\Dec(f)$,
  $\Dec(f)(\dec{p}) = \dec{\mkern-3mu\dec{p}f}$. Picking some $\dec{g} \in
  \Dec(X)$ we must have $(!+!) \dec{\mkern-3mu\dec{g} f} = (!+!) \dec{g} f$, 
  which indeed follows by the statement above. On the other hand, picking some 
  $p \in \Hom(X,1+1)$, chasing yields that we must have 
  $\dec{\mkern-3mu\dec{p}f} = \dec{pf}$, which follows directly by 
  Lemma~\ref{lem:utility} \eqref{lem:5}.
\end{proof}

A consequence of this equivalence in extensive restriction categories
is that decisions are a property of the \emph{predicates} rather than a
property of arbitrary maps, as it is commonly presented. This is shown in the following corollary to Theorem~\ref{thm:dec_cond_duality}.

\begin{corollary}
  A restriction category with restriction coproducts, a restriction zero, and a
  restriction terminal object has all decisions (\ie, is extensive as a
  restriction category) iff it has all decisions of predicates.
\end{corollary}
\begin{proof}
  It follows directly that having decisions for all morphisms implies having
  decisions for all predicates. On the other hand, suppose that the category
  only has decisions for predicates, and let $X \tot{f} Y+Z$ be an arbitrary
  morphism. But then, by the proof of Theorem~\ref{thm:dec_cond_duality}, the
  decision for the predicate $X \tot{f} Y+Z \tot{!+!} 1+1$ decides $X \tot{f} 
  Y+Z$ (by $\dec{(!+!)f} = \dec{(!+!)\dec{f}\mkern-3mu} = \dec{f}$), and we are 
  done.
\end{proof}

\section{The internal logic of extensive restriction categories} 
\label{sec:the_internal_logic_of_extensive_restriction_categories}
Having established the natural isomorphism of decisions and predicates (with
their respective transformers) which forms the condition/decision duality at
the categorical level, we now turn to their structure. The main result of this
section, Theorem~\ref{thm:internal_logic}, shows that the decisions $\Dec(X)$
on an object $X$ form a model of \KTw, and that decision transformers $\Dec(Y)
\tot{f^\diamond} \Dec(X)$ are homomorphisms of these models. We first recall Kleene's three valued logics, in particular \KTw{} and its algebraic counterpart of De Morgan quasilattices.

\subsection{Kleene's three valued logics and De Morgan quasilattices} 
\label{sub:weak_kleene_logic_and_de_morgan_quasilattices} 
Kleene's three valued logics of \KT{} (\emph{strong Kleene logic}) and \KTw{} 
(\emph{weak Kleene logic}), both introduced in~\cite{Kleene1952}, are logics
based on \emph{partial predicates} with a computational interpretation:
Predicates are conceived of as programs which \emph{may not terminate}, but if
they do, they terminate with a Boolean truth value as output. In this way, both
\KT{} and \KTw{} can be thought of as partial versions of classical logic.
Here, possible nontermination is handled analogously to how it is handled in
domain theory, \ie, by the introduction of a third truth value in addition to
truth $t$ and falsehood $f$, denoted $u$ in Kleene's
presentation~\cite{Kleene1952}, which should be read as ``undefined''.

The difference between \KT{} and \KTw{} lies in how they cope with undefined
truth values. In \KTw{} (see Figure~\ref{fig:k3w_semantics}), undefinedness is
``contagious'': if any part of an expression is undefined, the truth value of
the entire expression is undefined as well\footnote{This contagious behaviour
has also been used to explain other phenomena. In philosophy,
\KTw{} is better known as \BT{} or \emph{Bochvar's nonsense logic} (see, \eg,
\cite{FinnGrigolia1993}), and the third truth value read as ``meaningless'' or
``nonsensical'' rather than ``undefined''. The central idea is that nonsense is
contagious: \eg, ``$2+2=5$ and gobbledygook'' is nonsensical even if part of it
can be assigned meaning.}. This fits well into a computation paradigm with
possible nontermination and only sequential composition available. In contrast,
the semantics of \KT{} is to try to recover definite truth values whenever
possible, even if part of the computation fails to terminate. For example, in
\KT{} (and unlike \KTw), $p \land q$ is considered false if one of $p$ and $q$ 
is false, even if the other is undefined. While this allows for some recovery
in the face of nontermination, computationally it seems to require parallel
processing capabilities.
\begin{figure}
\centering
\subfloat[Weak conjunction.]{
\begin{tabular}{cc|ccc}
  & $P$ & $t$ & $f$ & $u$ \\
  $Q$ & $\land$ & & & \\
  \hline
  $t$ & & $t$ & $f$ & $u$ \\ 
  $f$ & & $f$ & $f$ & $u$ \\ 
  $u$ & & $u$ & $u$ & $u$
\end{tabular}} \hspace{5em}
\subfloat[Weak disjunction.]{
\begin{tabular}{cc|ccc}
  & $P$ & $t$ & $f$ & $u$ \\
  $Q$ & $\lor$ & & & \\
  \hline
  $t$ & & $t$ & $t$ & $u$ \\ 
  $f$ & & $t$ & $f$ & $u$ \\ 
  $u$ & & $u$ & $u$ & $u$
\end{tabular}} \hspace{5em}
\subfloat[Negation.]{
\begin{tabular}{c|ccc}
  $P$ & $t$ & $f$ & $u$ \\
  \hline
  $\neg P$ & $f$ & $t$ & $u$
\end{tabular}}
\caption{The three-valued semantics of \KTw.}
\label{fig:k3w_semantics}
\end{figure}

Like classical logic takes its algebraic semantics in Boolean algebras, the
corresponding algebraic structure for \KTw{} is that of \emph{De Morgan
quasilattices} (see, \eg, \cite{FinnGrigolia1993}). As is sometimes done, we
assume these to be distributive; \ie, what we call De Morgan quasilattices are
sometimes called \emph{distributive De Morgan quasilattices} or even
\emph{(distributive) De Morgan bisemilattices} (see, \eg, \cite{Ledda2018}). Note that we generally do
\emph{not} require these to be bounded, \ie, for top and bottom elements $\top$
and $\bot$ to exist.

\begin{definition}\label{def:dmqlat}
  A \emph{De Morgan quasilattice} (in its algebraic formulation) is a quadruple
  $\mathfrak{A} = (|\mathfrak{A}|, \neg, \land, \lor)$ satisfying the following
  equations, for all $p,q,r \in |\mathfrak{A}|$:
  \begin{multicols}{2}
    \begin{enumerate}[(i)]
      \item $p \land p = p$, \label{def:con_idemp}
      \item $p \lor p = p$, \label{def:dis_idemp}
      \item $p \land q = q \land p$, \label{def:con_comm}
      \item $p \lor q = q \lor p$, \label{def:dis_comm}
      \item $p \land (q \land r) = (p \land q) \land r$, \label{def:con_assoc}
      \item $p \lor (q \lor r) = (p \lor q) \lor r$, \label{def:dis_assoc}
      \item $p \land (q \lor r) = (p \land q) \lor (p \land r)$, 
      \label{def:con_dist}
      \item $p \lor (q \land r) = (p \lor q) \land (p \lor r)$,
      \label{def:dis_dist}
      \item $\neg\neg p = p$, \label{def:nne}
      \item $\neg(p \land q) = \neg p \lor \neg q$, \label{def:dm1}
      \item $\neg(p \lor q) = \neg p \land \neg q$, \label{def:dm2}
    \end{enumerate}
  \end{multicols}
  Further, a De Morgan quasilattice $\mathfrak{A}$ is said to be \emph{bounded}
  if there exist elements $\bot, \top \in |\mathfrak{A}|$ such that the
  following are satisfied (for all $p \in |\mathfrak{A}|$):
  \begin{multicols}{2}
    \begin{enumerate}[(i)]\setcounter{enumi}{11}
      \item $p \land \top = p$, \label{def:con_unit} and
      \item $p \lor \bot = p$. \label{def:dis_unit}
    \end{enumerate}
  \end{multicols}
  A homomorphism $\mathfrak{A} \tot{h} \mathfrak{B}$ of De Morgan quasilattices
  is a function $|\mathfrak{A}| \to |\mathfrak{B}|$ which preserves $\neg$,
  $\land$, and $\lor$. A homomorphism of bounded De Morgan quasilattices is one
  which additionally preserves $\top$ and $\bot$.
\end{definition}
Being a De Morgan quasilattice is a strictly weaker property than being a
Boolean algebra. In particular, a Boolean algebra is a bounded De Morgan
quasilattice which further satisfies the \emph{absorption} laws $p =
p \land (p \lor q)$ and $p = p \lor (p \land q)$, and the laws of contradiction and \emph{tertium non datur}, $p \land \neg p = \bot$ and $p \lor \neg p = \top$.

De Morgan quasilattices and their homomorphisms form a category which we call
\DMQLat. As for Boolean algebras, one can derive a partial order on De Morgan
quasilattices by $p \preccurlyeq q$ iff $p \land q = p$, and another one by $p
\dle q$ iff $p \lor q = q$. Unlike as for Boolean algebras, however, these do
\emph{not} coincide, though they are anti-isomorphic, as it follows from the De
Morgan laws that $p \preccurlyeq q$ iff $\neg q \dle \neg p$. We will return to
these in
Section~\ref{sec:the_internal_logic_of_extensive_restriction_categories} and
argue why $\cdot \preccurlyeq \cdot$ is the one more suitable as the entailment
relation for \KTw. 

\subsection{The internal logic} 
\label{sub:the_internal_logic}
With \KTw{} and De Morgan quasilattices introduced, we return to the construction of the internal logic. To aid in its presentation (and
subsequent proofs), we start by introducing a graphical language of extensive
restriction categories, based on the one for cocartesian categories (see,
\eg, \cite{Selinger2011}). Then, we show how the constants and connectives of
\KTw{} can be interpreted (Definition~\ref{def:internal_log}) as decisions on
an object (Lemma~\ref{lem:condec}). Finally, we show that decisions on an
object form a model of \KTw{} (Lemma~\ref{lem:dec_dmq}), and that decision
transformations are homomorphisms of these models
(Lemma~\ref{lem:f_diam_homo}), concluding this construction. We go on to
explore an important corollary to this construction, namely that if we restrict
ourselves from ordinary decisions to total decisions and total decision
transformations, we obtain a fibration over Boolean algebras instead
(Corollary~\ref{cor:tdec} Theorem~\ref{thm:ext_bool}). The latter is a
well-known property of extensive categories first shown in \cite{CJWW2015},
though this proof uses entirely different machinery.

Figure~\ref{fig:graphical_lang} shows the graphical language of extensive
restriction categories, which has the restriction coproduct as its monoidal tensor. The first five gadgets are from cocartesian categories
($\gamma_{X,Y}$ is here the twist map, $[\kappa_2, \kappa_1]$). We add gadgets
corresponding to decisions $X \tot{\dec{f}} X+X$, inverses to decisions $X+X
\tot{\dec{f}^\dagger} X$ (as all decisions are partial isomorphisms, see
Lemma~\ref{lem:utility}\eqref{lem:2}), and restriction idempotents $X
\tot{\ridm{f}} X$. The gadget for inverses to decisions was inspired by \emph{assertions} in reversible flowcharts (see~\cite{YokoyamaAG2016}). Useful derived gadgets include
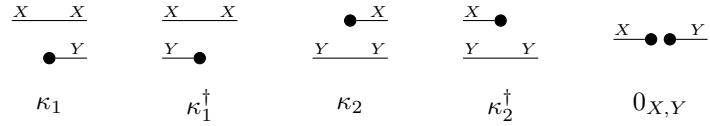
\ctikzfig{graphical_derived}
Just as the graphical language of cocartesian categories, isomorphism or
isotopy of diagrams is not enough for coherence -- equations only hold in the
graphical language up to diagrammatic manipulations corresponding to the
decision laws, as well as the diagrammatic manipulations for coproducts (\eg,
the commutative monoid axioms and naturality for the codiagonal, the zero morphism laws, etc.). For
more on the latter, see \cite{Selinger2011}. For example, graphically, the
decision laws are
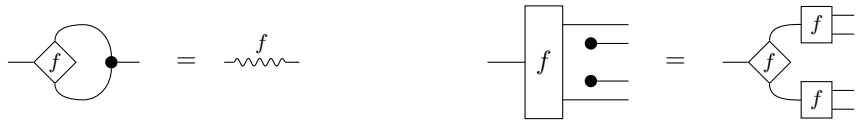
\ctikzfig{decision_laws}
As in the example above, when the signature is clear from the context, we omit
the object annotations (\eg, $X, Y, Z$ in Figure~\ref{fig:graphical_lang}).

\begin{figure}
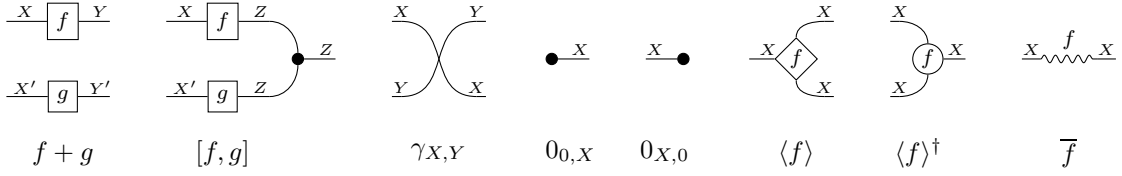

\ctikzfig{graphical_language}
\caption{An overview of the gadgets that make up the graphical language of
extensive restriction categories.}
\label{fig:graphical_lang}
\end{figure}

With the graphical language in place, we proceed to give the definition of the
internal logic of decisions in an extensive restriction category, \ie, the
entailment relation and construction of constants and propositional connectives.

\begin{definition}\label{def:internal_log}
  In an extensive restriction category, propositional constants and connectives 
  are defined for decisions as follows, using the graphical language:
  \ctikzfig{logic}
  Entailment is defined by $\dec{p} \entails \dec{q}$ iff $\dec{p}
  \preccurlyeq \dec{q}$ (explicitly, iff $\dec{p} \land \dec{q} = \dec{p}$). 
\end{definition}

For those more textually inclined, this defines $\top = \kappa_1$, $\bot =
\kappa_2$, $\neg \dec{p} = \gamma \dec{p}$, $\dec{p} \lor \dec{q} =
(\dec{p}^\dagger + \id) \alpha (\ridm{\dec{q}} + \dec{q})\dec{p}$, and $\dec{p}
\land \dec{q} = (\id + \dec{p}^\dagger)\alpha(\dec{q} + \ridm{\dec{q}})\dec{p}$.

Intuitively, we think of decisions as representing partial predicates by separating values into \emph{witnesses} and \emph{counterexamples} of that
partial predicate (see also \cite{KaarsgaardGlueck2018}). The definitions of
$\top$ and $\bot$ express the convention that the first component carries
witnesses, while the second component carries counterexamples. Negation of
partial predicates then amounts to swapping witnesses for counterexamples and
vice versa, \ie, by composing with the symmetry. The intuition behind
conjunction (and, dually, disjunction) is less obvious: Using the intuition of
decisions as morphisms that tag inputs with a branch but doesn't change it
otherwise, we see that a witness of $\dec{p} \land \dec{q}$ has to be a witness
of both $\dec{p}$ and $\dec{q}$, while a counterexample of $\dec{p} \land
\dec{q}$ is either a counterexample of $\dec{p}$ which is further defined for
$\dec{q}$ (necessary to ensure commutativity), or a witness of $\dec{p}$ which
is a counterexample of $\dec{q}$. The case for disjunctions is dual.

Before we move on to show that this actually has the logical structure we're
after, we first obliged to show that these connectives and constants actually define well-formed decisions. This fact is expressed in the following lemma. 
\begin{lemma}\label{lem:condec}
  The constants and connectives of Definition~\ref{def:internal_log} are 
  decisions.
\end{lemma}
\ifappendix
\begin{proof}
  See appendix.
\end{proof}
\fi
Before we can proceed, we need a small technical lemma.
\begin{lemma}\label{lem:dec_misc}
  Let \begin{minipage}[m]{8mm}\input{pdec.tikz}\end{minipage} and 
  \begin{minipage}[m]{8mm}\input{qdec.tikz}\end{minipage} be decisions. It is the 
  case that
  \begin{multicols}{2}
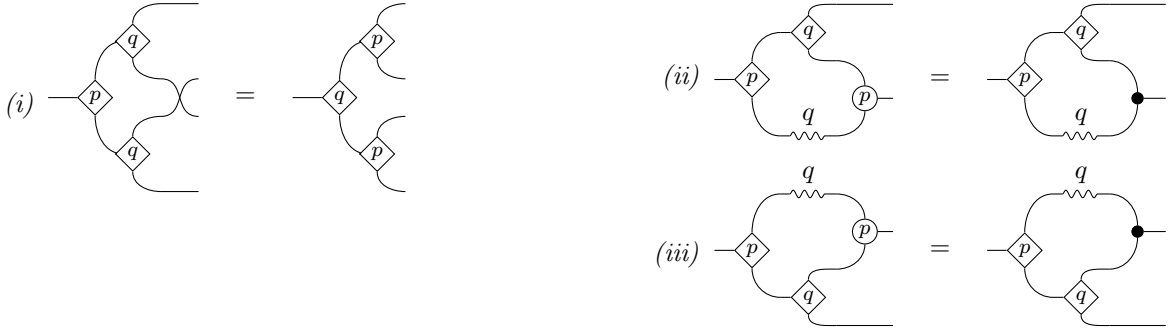
  \begin{enumerate}[(i)]
    \item \begin{minipage}[m]{\linewidth}\input{commstmt.tikz}\end{minipage} 
    \label{lem:commstmt}
    \item \begin{minipage}[m]{\linewidth}\input{con_irrev.tikz}\end{minipage} 
    \label{lem:con_irrev}
    \item \begin{minipage}[m]{\linewidth}\input{dis_irrev.tikz}\end{minipage} 
    \label{lem:dis_irrev}
  \end{enumerate}
  \end{multicols}  
\end{lemma}
\ifappendix
\begin{proof}
  See appendix.
\end{proof}
\fi
The first part of this lemma can be seen as a form of commutativity for
decisions -- and, indeed, it performs most of the heavy lifting in showing
commutativity of conjunction and disjunction. On the other hand, parts
\eqref{lem:con_irrev} and \eqref{lem:dis_irrev} shows that we could have
defined conjunction and disjunction more simply in
Definition~\ref{def:internal_log}. The reason why we chose the current
definition is that it yields entirely \emph{reversible} models (see also
\cite{KaarsgaardGlueck2018}), \ie, involving only partial isomorphisms. We will
discuss this property further in Section~\ref{sec:conclusion}. For now, we continue with the internal logic.

\begin{lemma}\label{lem:dec_dmq}
  $\Dec(X)$ is a bounded De Morgan quasilattice for any object $X$.
\end{lemma}
\begin{proof}
  \ifappendix
  We show only a few of the cases here using the graphical language. See the
  appendix for the rest.
  \fi
  Idempotence of conjunction, \ie, $\dec{p} \land \dec{p} = \dec{p}$, follows by
  \ctikzfig{con_idemp}
  and similarly for disjunction. That $\dec{p} \land \top = \dec{p}$ is shown 
  simply by 
  \ctikzfig{con_unit}
  and again, the unit law for disjunction has an analogous proof. The first De 
  Morgan law, that $\neg \dec{p} \land \neg \dec{q} = \neg (\dec{p} \lor 
  \dec{q})$
  \ctikzfig{dm1}
  and the proof of the second De Morgan law follows similarly.
\end{proof}
As such, we have that each collection of decisions on an object form a local
model of \KTw, giving us the first part of the fibration. For the second, we
need to show that decision transformers preserve entailment and the
propositional connectives (though not necessarily the constants). This is shown
in the following lemma.

\begin{lemma}\label{lem:f_diam_homo}
  Let $X \tot{f} Y$. Then $\Dec(Y) \tot{f^\diamond} \Dec(X)$ is a homomorphism
  of De Morgan quasilattices, \ie,
  \begin{multicols}{2}
  \begin{enumerate}[(i)]
    \item\label{lem:monot} $\dec{p} \entails \dec{q}$ implies 
    $f^\diamond(\dec{p}) \entails 
    f^\diamond(\dec{q})$
    \item\label{lem:pres_neg} $f^\diamond(\neg\dec{p}) = \neg 
    f^\diamond(\dec{p})$
    \item\label{lem:pres_conj} $f^\diamond(\dec{p} \land \dec{q}) =
    f^\diamond(\dec{p}) \land f^\diamond(\dec{q})$
    \item\label{lem:pres_disj} $f^\diamond(\dec{p} \lor \dec{q}) =
    f^\diamond(\dec{p}) \lor f^\diamond(\dec{q})$
  \end{enumerate}
  \end{multicols}
  In addition, if $f$ is total then $f^\diamond$ is a homomorphism of
  \emph{bounded} De Morgan quasilattice; \ie, we also have $f^\diamond(\top) =
  \top$ and $f^\diamond(\bot) = \bot$.
\end{lemma}
\begin{proof}
  \eqref{lem:monot} follows by \eqref{lem:pres_conj} since 
  $\dec{p} \entails \dec{q}$ iff $\dec{p} \preccurlyeq \dec{q}$ iff $\dec{p}
  \land \dec{q} = \dec{q}$, which in turn implies that $f^\diamond(\dec{p}) 
  \land f^\diamond(\dec{q}) = f^\diamond(\dec{p} \land \dec{q}) = 
  f^\diamond(\dec{p})$, so $f^\diamond(\dec{p}) \preccurlyeq 
  f^\diamond(\dec{q})$ as well, \ie, $f^\diamond(\dec{p}) \entails 
  f^\diamond(\dec{q})$.

  For \eqref{lem:pres_neg}, we compute $f^\diamond(\neg\dec{p}) =
  f^\diamond(\gamma\dec{p}) = f^\diamond(\dec{\gamma p}) = 
  \dec{\mkern-3mu\dec{\gamma p}f} = \dec{\gamma pf} = \gamma \dec{pf} = \gamma 
  \dec{\mkern-3mu\dec{p}f} = \neg f^\diamond(\dec{p})$ (using 
  Lemma~\ref{lem:utility}).
  
  \eqref{lem:pres_conj} follows by lengthy but straightforward 
  computation\ifappendix (see appendix)\fi.
  
  \eqref{lem:pres_disj} is analogous to the previous case.
\end{proof}

Notice the final part regarding preservation of units. Generally,
$f^\diamond(\top) = \dec{\top f} = \dec{\kappa_1 f}$, so
$\ridm{f^\diamond(\top)} = \ridm{\dec{\kappa_1 f}} = \ridm{\kappa_1 f} =
\ridm{\ridm{\kappa_1} f} = \ridm{f}$, so if $f$ is not total, $\dec{\kappa_1 f}
\neq \kappa_1$ (instead $\dec{\kappa_1 f} = \kappa_1 \ridm{f}$). 

Putting the two lemmas together gives us the main result:

\begin{theorem}\label{thm:internal_logic}
  In every extensive restriction category \C, decisions over \C{} form a fibred 
  De Morgan quasilattice via the decision fibration.
\end{theorem}
\begin{proof}
  By Lemmas~\ref{lem:dec_dmq} and \ref{lem:f_diam_homo}.
\end{proof}
We previously claimed that the conjunction order was the more suitable one for
entailment in extensive restriction categories. We are finally ready to state
why:

\begin{lemma}\label{lem:entailment}
  Entailment is upwards directed in truth and definedness: $\dec{p} \entails
  \dec{q}$ iff $\kappa_1^\dagger \dec{p}
  \le \kappa_1^\dagger \dec{q}~\text{and}~\ridm{\dec{p}} \le \ridm{\dec{q}}$.
\end{lemma}
\ifappendix
\begin{proof}
  See appendix.
\end{proof}
\fi

In other words, $\dec{p}$ entails $\dec{q}$ iff $\dec{q}$ is both \emph{at
least as true} and \emph{at least as defined} as $\dec{p}$ is. That is,
entailment preserves not only \emph{truth} (as we expect all entailments to)
but also \emph{information} (as we expect of orders on \emph{partial} maps).
Compare this to the disjunction partial order for which $\dec{p} \dle \dec{q}$
instead states that $\dec{q}$ is less \emph{false} and less \emph{defined} than
$\dec{p}$: In other words, it prefers for information to be \emph{forgotten}
rather than preserved.

We move on now to an important special case of the situation above, which is
when only \emph{total} decisions are considered rather than arbitrary ones. For
this, we need a small lemma regarding the restriction idempotents of decisions
when composed using the propositional connectives.

\begin{lemma}\label{lem:ridmdec}
  We state some facts about restriction idempotents of decisions:
  \begin{multicols}{2}
  \begin{enumerate}[(i)]
    \item\label{lem:ridmdec:1} $\ridm{\neg\dec{p}} = \ridm{\dec{p}}$,
    \item\label{lem:ridmdec:2} $\ridm{\dec{p} \land \dec{q}} =
    \ridm{\dec{p}}\,\ridm{\dec{q}}$,
    \item\label{lem:ridmdec:3} $\ridm{\dec{p} \lor \dec{q}} =
    \ridm{\dec{p}}\,\ridm{\dec{q}}$,
    \item\label{lem:ridmdec:4} $\ridm{\dec{p} \land \dec{q}} \le
    \ridm{\dec{p}}$ and $\ridm{\dec{p} \land \dec{q}} \le \ridm{\dec{q}}$,
    \item\label{lem:ridmdec:5} $\ridm{\dec{p} \lor \dec{q}} \le \ridm{\dec{p}}$
    and $\ridm{\dec{p} \lor \dec{q}} \le \ridm{\dec{q}}$.
  \end{enumerate}
  \end{multicols}
\end{lemma}
\ifappendix
\begin{proof}
  See appendix.
\end{proof}
\fi
We can now show that total decisions form a fibred Boolean algebra.
\begin{corollary}\label{cor:tdec}
  $\TDec(X)$ is a Boolean algebra for any object $X$, and $f^\vartriangle :
  \TDec(Y) \to \TDec(X)$ is homomorphism of Boolean algebras for any total 
  $X \tot{f} Y$.
\end{corollary}
\begin{proof}
  Since $\Dec(X)$ is a De Morgan quasilattice (Lemma~\ref{lem:dec_dmq}), since
  total decisions are specifically decisions, and since the constants are total
  and the connectives preserve totality (Lemma~\ref{lem:ridmdec}), it suffices
  to show that when $\dec{p}$ and $\dec{q}$ are total they satisfy the
  \emph{absorption} laws $\dec{p} = \dec{p} \land (\dec{p} \lor \dec{q}$ and
  $\dec{p} = \dec{p} \lor (\dec{p} \land \dec{q})$, and the laws of
  contradiction and \emph{tertium non datur}, $\dec{p} \land \neg\dec{p} = 
  \bot$ and $\dec{p} \lor \neg\dec{p} = \top$. The first absorption law
  follows by
  \ctikzfig{absorption}
  and the other follows analogously. Likewise, the law of contradiction can be 
  shown as
  \ctikzfig{contradiction}
  and similarly for \emph{tertium non datur}.
\end{proof}

Using the previous corollary, it follows (see \cite{CJWW2015} for the original
proofs from effectus theory) that predicates over an extensive category form a
fibred Boolean algebra.
\begin{theorem}\label{thm:ext_bool}
  Predicates (or, equivalently, decisions) over an extensive category is a
  fibred Boolean algebra via the predicate fibration (or, equivalently, the
  decision fibration).
\end{theorem}
\begin{proof}
  Since total decisions on objects form Boolean algebras by
  Corollary~\ref{cor:tdec}, it suffices to show that every extensive category
  arises as the subcategory of total morphisms of an extensive restriction
  category.
  
  Let \C{} be an extensive category, and $\mathcal{M}$ denote the collection of 
  all coproduct injections of \C{}. As remarked in \cite{Cockett2002}, this is 
  a stable system of monics, and by Example 4.17 of \cite{Cockett2003}, 
  $\Par(\C,\mathcal{M})$ is a classified restriction category under the $+1$ 
  monad. Since \C{} has coproducts and $\Par(\C,\mathcal{M})$ is classified, it 
  follows by Proposition 2.3 of \cite{Cockett2007} that $\Par(\C,\mathcal{M})$
  has restriction coproducts. That $0$ is a restriction zero in 
  $\Par(\C,\mathcal{M})$ follows straightforwardly, with the span 
  $X \fromm{!_X} 0 \tot{!_Y} Y$ as the unique zero morphism $X \tot{0_{X,Y}}
  Y$. As such, it suffices to show that decisions can be constructed in
  $\Par(\C,\mathcal{M})$. Let $X \fromm{m} X' \tot{f} Y+Z$ be an arbitrary
  morphism of $\Par(\C,\mathcal{M})$. Since \C{} is extensive it has pullbacks
  of coproduct injections along arbitrary morphisms, so the two squares
  \begin{center}
    \begin{tikzcd}
    X_1 \arrow[r, "m_1"] \arrow[d, "f_1"'] & X' \arrow[d, "f"] & X_2 \arrow[l, "m_2"'] \arrow[d, "f_2"] \\
    Y \arrow[r, "\kappa_1"']               & Y+Z               & Z \arrow[l, "\kappa_2"]               
    \end{tikzcd}
  \end{center}
  are pullbacks, and so the top row is a coproduct diagram (\ie, $X_1 + X_2
  \cong X'$). But then it readily follows that
  \begin{center}
    \begin{tikzcd}[cramped, sep=tiny, font=\scriptsize]
      &                     & X_1+X_2 \arrow[rd, "m_1+m_2"] \arrow[ld, "\cong"'] &                         &     \\
      & X' \arrow[ld, "m"'] &                                                    & X'+X' \arrow[rd, "m+m"] &     \\
    X &                     &                                                    &                         & X+X
    \end{tikzcd}
  \end{center}
  is a decision for $X \fromm{m} X' \tot{f} Y+Z$ in $\Par(\C,\mathcal{M})$, 
  and we are done.
\end{proof}


\section{Conclusion and future work} 
\label{sec:conclusion}
Motivated by an observation from flowchart languages that predicates serve a
dual role as both condition and decision, we have given an account of extensive
restriction categories (due to \cite{Cockett2002,Cockett2003,Cockett2007}) as
categories with an internal logic (namely \KTw) that internalize this duality,
in the form of a natural isomorphism between the predicate fibration and the
decision fibration.

We have also extended the graphical language of cocartesian
categories to one for extensive restriction categories, and used our results to
give an alternative proof of the fact that extensive categories, too, are
categories with an internal logic -- classical logic.
While the graphical language has proven itself useful in proving theorems, it
does have its shortcomings. For example, the only way to express restriction
idempotents of compositions, such as $\ridm{gf}$, is, awkwardly, as
\input{ridmfg.tikz}. That is, we would want only one representation of composition
as placing gadgets in sequence, but since $\ridm{gf}$ cannot generally be
expressed as a composite involving only smaller things (\eg, $\ridm{f}$ and
$\ridm{g}$), we are forced in this case to let the textual representation (\ie,
juxtaposition) bleed into the graphical language. The graphical notation for
decisions has similar issues.

An application of the developed theory is in reversible models of logics, which
was also the motivation for defining the connectives in slightly more involved
fashion, using partial inverses to decisions rather than the codiagonal.
Indeed, the inspiration for using decisions as predicates came from the study
of the categorical semantics of reversible flowchart languages (see
\cite{Glueck2017,KaarsgaardGlueck2018}). Since a decision in \C{} is still a
decision in $\Inv(\C)$ (see \cite{KaarsgaardGlueck2018}), $\Dec(X)$ is still a
De Morgan quasilattice in $\Inv(\C)$, though the homomorphisms between fibres
differ (\ie, only decision transformers that are partial isomorphisms occur in
the decision fibration on $\Inv(\C)$).

We have only considered the weak Kleene logic \KTw{} here, as it can be
constructed by purely sequential means. However, we conjecture that the strong
Kleene logic \KT{} can be modelled as well in extensive restriction categories
if additionally a \emph{parallel} composition operator such as finite joins
(see \cite{Guo2012}) is available. Finally, just the propositional fragment of
\KTw{} and classical logic has been considered in this paper. Though decisions
on an object yields a fibred category with a logical structure, we have not
explored extensions to models of first-order logics, \eg, by investigating the
feasibility of adjoints to substitution, as in the standard trick due to
Lawvere~\cite{Lawvere1969} (see also \cite{Jacobs1999}).

\bibliography{mfps}
\bibliographystyle{abbrv}

\ifappendix

\newpage

\appendix

\section{Omitted proofs} 
\label{sec:omitted_proofs}
\begin{proof}[Proof of Lemma~\ref{lem:utility}]
  For \eqref{lem:1} and \eqref{lem:2}, see \cite{Cockett2007}. \eqref{lem:2.5}
  follows by \eqref{lem:2} since
  $$\ridm{\dec{f}^\dagger} = \ridm{[\ridm{\kappa_1^\dagger
  f},\ridm{\kappa_2^\dagger f}]} = \ridm{\ridm{\kappa_1^\dagger f} +
  \ridm{\kappa_2^\dagger f}} = \ridm{\kappa_1^\dagger f}+\ridm{\kappa_2^\dagger
  f}.$$
  For \eqref{lem:3},
  $\ridm{\dec{f}} = \ridm{\dec{f}} = \ridm{\ridm{\nabla} \dec{f}} =
  \ridm{\nabla \dec{f}} = \ridm{\ridm{f}} = \ridm{f}$. To show \eqref{lem:4} we
  show that $\gamma \dec{f}$ decides $\gamma f$, since $\nabla \gamma \dec{f} =
  \nabla \dec{f} = \ridm{f} = \ridm{f} = \ridm{\ridm{\gamma} f} =
  \ridm{\gamma f}$ and
  \begin{align*}
    ((\gamma f) + (\gamma f)) \gamma \dec{f} 
    & = \gamma ((\gamma f) + (\gamma f)) \dec{f}
     = \gamma (\gamma + \gamma) (f+f) \dec{f}
     = \gamma (\gamma + \gamma) (\kappa_1 + \kappa_2) p \\
    & = \gamma ((\gamma \kappa_1) + (\gamma \kappa_2)) p
     = \gamma (\kappa_2 + \kappa_1) p
     = (\kappa_1 + \kappa_2) \gamma p
  \end{align*}
  which was what we wanted. We show \eqref{lem:5} analogously by showing that
  $\dec{\mkern-3mu\dec{f}g}$ decides $\dec{fg}$ since
  $\nabla\dec{\mkern-3mu\dec{f}g} = \ridm{\dec{f}g} = \ridm{\ridm{\nabla} 
  \dec{f}g} = \ridm{\nabla \dec{f} g} = \ridm{\ridm{f}g} = \ridm{fg}$ and
  \begin{align*}
    ((fg)+(fg))\dec{\mkern-3mu\dec{f}g} 
    & = ((f\ridm{f}g)+(f\ridm{f}g))\dec{\mkern-3mu\dec{f}g}
     = ((f\nabla\dec{f}g)+(f\nabla\dec{f}g))\dec{\mkern-3mu\dec{f}g} \\
    & = ((f\nabla)+(f\nabla)) ((\dec{f}g) + (\dec{f}g))\dec{\mkern-3mu\dec{f}g} 
    \\
    & = ((f\nabla)+(f\nabla)) (\kappa_1 + \kappa_2)\dec{f}g
     = ((f\nabla\kappa_1)+(f\nabla\kappa_2)) \dec{f}g \\
    & = (f+f) \dec{f}g
    = (\kappa_1+\kappa_2) fg \enspace.
  \end{align*}
  For \eqref{lem:6lem}, we observe that 
  $\ridm{(\ridm{e}+\ridm{e})\dec{f}} =
  \ridm{\nabla (\ridm{e}+\ridm{e})\dec{f}} = 
  \ridm{\ridm{e}\nabla\dec{f}} =
  \ridm{\ridm{e}\ridm{f}}  = \ridm{e}\ridm{f} = \ridm{f}\ridm{e} =
  \ridm{\dec{f}}\ridm{e}$
  so $ (\ridm{e}+\ridm{e})\dec{f} = (\ridm{e}+\ridm{e})\dec{f}
  \ridm{(\ridm{e}+\ridm{e})\dec{f}} = (\ridm{e}+\ridm{e})\dec{f} \ridm{\dec{f}}
  \ridm{e} = (\ridm{e}+\ridm{e})\dec{f} \ridm{e}. $
  To show \eqref{lem:6}, we show that the two morphism decide one another. We
  see that $\dec{f}\ridm{e}$ decides $(\ridm{e}+\ridm{e})\dec{f}$ since
  $\nabla\dec{f}\ridm{e} = \ridm{f}\ridm{e} = \ridm{\dec{f}}\ridm{e} =
  (\ridm{e}+\ridm{e})\dec{f}$ (see \eqref{lem:6lem} above) and
  \begin{align*}
    & (((\ridm{e}+\ridm{e})\dec{f})+((\ridm{e}+\ridm{e})\dec{f})) 
    \dec{f}\ridm{e}
     = ((\ridm{e}+\ridm{e})+(\ridm{e}+\ridm{e}))(\dec{f}+\dec{f})
    \dec{f}\ridm{e} \\
    & \quad = ((\ridm{e}+\ridm{e})+(\ridm{e}+\ridm{e}))(\dec{f}+\dec{f})
    \dec{\mkern-3mu\dec{f}\mkern-3mu}\ridm{e}
    = ((\ridm{e}+\ridm{e})+(\ridm{e}+\ridm{e}))(\kappa_1+\kappa_2)
    \dec{f}\ridm{e} \\
    & \quad = (\kappa_1+\kappa_2) (\ridm{e}+\ridm{e}) \dec{f}\ridm{e}
    = (\kappa_1+\kappa_2) (\ridm{e}+\ridm{e}) \dec{f}
  \end{align*}
  where $(\ridm{e}+\ridm{e}) \dec{f} = (\ridm{e}+\ridm{e}) \dec{f} \ridm{e}$ by
  \eqref{lem:6lem}; thus $\dec{f} \ridm{e}$ decides $(\ridm{e}+\ridm{e})
  \dec{f}$, \ie, $\dec{\mkern-2mu(\ridm{e}+\ridm{e}) \dec{f}\mkern-3mu} =
  \dec{f} \ridm{e} = \dec{\mkern-3mu\dec{f} \ridm{e}}$ (the latter by
  \eqref{lem:1}). In the other direction, $\nabla (\ridm{e}+\ridm{e}) \dec{f} = 
  \ridm{e} \nabla \dec{f} = \ridm{e} \ridm{f} = \ridm{f} \ridm{e} = 
  \ridm{\dec{f}} \ridm{e}$ and it is the case that
  \begin{align*}
    & ((\dec{f}\ridm{e})+(\dec{f}\ridm{e}))(\ridm{e}+\ridm{e})\dec{f} = ((\dec{f}\ridm{e})+(\dec{f}\ridm{e}))(\ridm{e}+\ridm{e})\dec{f}
    \ridm{e} \\
    & \quad = ((\dec{f}\ridm{e}\ridm{e})+(\dec{f}\ridm{e}\ridm{e}))\dec{f}
    \ridm{e} = ((\dec{f}\ridm{e})+(\dec{f}\ridm{e}))\dec{f} \ridm{e} \\
    & \quad = ((\dec{f}\ridm{e})+(\dec{f}\ridm{e}))\dec{\mkern-3mu\dec{f} 
    \ridm{e}} = (\kappa_1 + \kappa_2)\dec{f} \ridm{e}
  \end{align*}
  
  For \eqref{lem:7} we have that $\dec{f}\ridm{e} =
  \dec{\mkern-3mu\dec{f}\ridm{e}}$ by \eqref{lem:6} and get $\dec{f}\ridm{e} =
  \dec{\mkern-3mu\dec{f}\ridm{e}} = \dec{f\ridm{e}}$ by \eqref{lem:5}.
  
  For \eqref{lem:8}, by \eqref{lem:2} $\dec{f}^\dagger =
  \left[\,\ridm{\kappa_1^\dagger f},\ridm{\kappa_2^\dagger f}\,\right]$ so
  $\kappa_i^\dagger \dec{f} = \kappa_i^\dagger \left[\,\ridm{\kappa_1^\dagger
  f}, \ridm{\kappa_2^\dagger f}\,\right]^\dagger = \left(\left[\,
  \ridm{\kappa_1^\dagger f}, \ridm{\kappa_2^\dagger f}\,\right]
  \kappa_i\right)^\dagger = \ridm{\kappa_i^\dagger f}^\dagger = 
  \ridm{\kappa_i^\dagger f}$.
  
  For \eqref{lem:9}, we compute
  \begin{align*}
    \dec{g}f & = (\kappa_1^\dagger + \kappa_2^\dagger) (\kappa_1 + \kappa_2) 
    \dec{g}f = (\kappa_1^\dagger + \kappa_2^\dagger) ((\dec{g}f) + (\dec{g}f))
    \dec{\mkern-3mu\dec{g}f} \\
    & = (\kappa_1^\dagger + \kappa_2^\dagger) ((\dec{g}f) + (\dec{g}f))
    \dec{gf} = (\kappa_1^\dagger + \kappa_2^\dagger) (\dec{g} + \dec{g}) (f+f)
    \dec{gf} \\
    & = ((\kappa_1^\dagger\dec{g}) + (\kappa_2^\dagger\dec{g})) (f+f)
    \dec{gf} = (\ridm{\kappa_1^\dagger g} + \ridm{\kappa_2^\dagger g}) (f+f)
    \dec{gf} \\
    & = (f+f) (\ridm{\kappa_1^\dagger gf} + \ridm{\kappa_2^\dagger gf})
    \dec{gf} 
    = (f+f) \ridm{[\ridm{\kappa_1^\dagger gf}, \ridm{\kappa_2^\dagger gf}]}
    \dec{gf} = (f+f) \ridm{\dec{gf}^\dagger} \dec{gf} \\
    & = (f+f) \dec{gf}
  \end{align*}
  where we use that $\kappa_i^\dagger\dec{g} = \ridm{\kappa_i^\dagger g}$ by
  \eqref{lem:9}, and $\dec{gf}^\dagger = [\ridm{\kappa_1^\dagger gf}, 
  \ridm{\kappa_2^\dagger gf}]$ by \eqref{lem:2}.
\end{proof}

\begin{proof}[Proof of Lemma~\ref{lem:condec}]
  That $\top = \kappa_1$ and $\bot = \kappa_2$ are decisions is shown in 
  \cite{Cockett2007}.
  That $\neg\dec{p} = \gamma\dec{p}$ is a decision follows by $\gamma\dec{p} =
  \dec{\gamma p}$ by Lemma~\ref{lem:utility}\eqref{lem:4}. To see that $\dec{p}
  \land \dec{q}$ is a decision, it suffices by Lemma~\ref{lem:decrep} to show 
  that $\kappa_1^\dagger (\dec{p} \land \dec{q}) = \ridm{\kappa_1^\dagger
  (\dec{p} \land \dec{q})}$ and $\kappa_2^\dagger (\dec{p} \land \dec{q}) = 
  \ridm{\kappa_2^\dagger (\dec{p} \land \dec{q})}$. We compute
  \begin{align*}
    \kappa_1^\dagger (\dec{p} \land \dec{q}) & =
    \kappa_1^\dagger (\id + \dec{p}^\dagger) \alpha (\dec{q} + \ridm{\dec{q}}) 
    \dec{p} \\
    & = \id \kappa_1^\dagger \alpha (\dec{q} + \ridm{\dec{q}}) \dec{p} \\
    & = \kappa_1^\dagger \alpha (\dec{q} + \ridm{\dec{q}}) \dec{p} \\
    & = \kappa_1^\dagger \kappa_1^\dagger (\dec{q} + \ridm{\dec{q}}) \dec{p} \\
    & = \kappa_1^\dagger \dec{q} \kappa_1^\dagger \dec{p} \\
    & = \ridm{\kappa_1^\dagger q}\, \ridm{\kappa_1^\dagger p}
  \end{align*}
  so $\kappa_1^\dagger (\dec{p} \land \dec{q}) = \ridm{\kappa_1^\dagger q}\,
  \ridm{\kappa_1^\dagger p} = \ridm{\ridm{\kappa_1^\dagger q}\,
  \ridm{\kappa_1^\dagger p}} = \ridm{\kappa_1^\dagger (\dec{p} \land \dec{q})}$.
  Further
  \begin{align*}
    \kappa_2^\dagger (\dec{p} \land \dec{q}) 
    & = \kappa_2^\dagger (\id + \dec{p}^\dagger) \alpha (\dec{q} +
    \ridm{\dec{q}}) \dec{p} \\
    & = \dec{p}^\dagger \kappa_2^\dagger \alpha (\dec{q} +
    \ridm{\dec{q}}) \dec{p} \\
    & = \dec{p}^\dagger (\kappa_2^\dagger + \id) (\dec{q} +
    \ridm{\dec{q}}) \dec{p} \\
    & = \dec{p}^\dagger ((\kappa_2^\dagger \dec{q}) +
    \ridm{\dec{q}})) \dec{p} \\
    & = \dec{p}^\dagger (\ridm{\kappa_2^\dagger q} +
    \ridm{q})) \dec{p} \\
    & \le \dec{p}^\dagger (\ridm{q} + \ridm{q})) \dec{p} \\
    & = \dec{p}^\dagger \dec{p} \ridm{q} \\
    & = \ridm{\dec{p}} \ridm{q} \\
    & = \ridm{p}\, \ridm{q}
  \end{align*}
  so since $\kappa_2^\dagger (\dec{p} \land \dec{q}) \le \ridm{p}\, \ridm{q}$
  it follows that $\kappa_2^\dagger (\dec{p} \land \dec{q}) =
  \ridm{\kappa_2^\dagger (\dec{p} \land \dec{q})}$, and so finally $\dec{p} 
  \land \dec{q}$ is a decision by Lemma~\ref{lem:decrep}. The case for $\dec{p} 
  \lor \dec{q}$ is entirely analogous.
\end{proof}

\begin{proof}[Proof of Lemma~\ref{lem:f_diam_homo}\eqref{lem:pres_conj}]
  By Lemma~\ref{lem:decrep} it suffices to show that $\kappa_1^\dagger
  f^\diamond(\dec{p} \land \dec{q}) = \kappa_1^\dagger (f^\diamond(\dec{p})
  \land f^\diamond(\dec{q}))$ and $\kappa_2^\dagger f^\diamond(\dec{p} \land
  \dec{q}) = \kappa_2^\dagger (f^\diamond(\dec{p}) \land f^\diamond(\dec{q}))$,
  Firstly we expand $f^\diamond(\dec{p} \land \dec{q}) = \dec{(\dec{p} \land
  \dec{q})f}$ and $f^\diamond(\dec{p}) \land f^\diamond(\dec{q}) =
  \dec{\mkern-3mu\dec{p}f} \land \dec{\mkern-3mu\dec{q}f} = \dec{pf} \land
  \dec{qf}.$ Then we compute
  \begin{align*}
    \kappa_1^\dagger (f^\diamond(\dec{p}) \land f^\diamond(\dec{q}) 
    & = \kappa_1^\dagger \dec{(\dec{p} \land \dec{q})f}
    = \ridm{\kappa_1^\dagger (\dec{p} \land \dec{q})f} \\
    & = \ridm{\kappa_1^\dagger (\id + \dec{p}^\dagger) \alpha (\dec{q} +
    \ridm{\dec{q}}) \dec{p} f} 
    = \ridm{\id \kappa_1^\dagger \alpha (\dec{q} +
    \ridm{\dec{q}}) \dec{p} f} \\
    & = \ridm{\kappa_1^\dagger \alpha (\dec{q} +
    \ridm{\dec{q}}) \dec{p} f} 
    = \ridm{\kappa_1^\dagger \dec{q} \kappa_1^\dagger \dec{p} f} \\
    & = \ridm{\ridm{\kappa_1^\dagger \dec{q}}\, \ridm{\kappa_1^\dagger \dec{p}} 
    f} 
    = \ridm{\ridm{\kappa_1^\dagger \dec{q}} f \ridm{\kappa_1^\dagger
    \dec{p}f}} \\
    & = \ridm{f \ridm{\kappa_1^\dagger \dec{q}f}\, \ridm{\kappa_1^\dagger
    \dec{p}f}} 
    = \ridm{\ridm{f}\, \ridm{\kappa_1^\dagger \dec{q}f}\,
    \ridm{\kappa_1^\dagger \dec{p}f}} \\
    & = \ridm{f}\, \ridm{\kappa_1^\dagger \dec{q}f}\,
    \ridm{\kappa_1^\dagger \dec{p}f} 
    = \ridm{\kappa_1^\dagger \dec{q}f} \ridm{f}\,
    \ridm{\kappa_1^\dagger \dec{p}f} \\
    & = \ridm{\kappa_1^\dagger \dec{q}f \ridm{f}}\,
    \ridm{\kappa_1^\dagger \dec{p}f} 
    = \ridm{\kappa_1^\dagger \dec{q}f}\,\ridm{\kappa_1^\dagger \dec{p}f}
  \end{align*}
  and
  \begin{align*}
    \kappa_1^\dagger (f^\diamond(\dec{p} \land f^\diamond(\dec{q})) 
    & = \kappa_1^\dagger \dec{pf} \land \dec{qf} 
    = \ridm{\kappa_1^\dagger (\dec{pf} \land \dec{qf})} \\
    & = \ridm{\kappa_1^\dagger (\id + \dec{pf}) \alpha 
    (\dec{qf}+\ridm{\dec{qf}}) \dec{pf}} 
    = \ridm{\id \kappa_1^\dagger \alpha (\dec{qf}+\ridm{\dec{qf}}) \dec{pf}}
    \\
    & = \ridm{\kappa_1^\dagger \alpha (\dec{qf}+\ridm{\dec{qf}}) \dec{pf}} 
    = \ridm{\kappa_1^\dagger \kappa_1^\dagger (\dec{qf}+\ridm{\dec{qf}}) 
    \dec{pf}} \\
    & = \ridm{\kappa_1^\dagger \dec{qf} \kappa_1^\dagger \dec{pf}} 
    = \ridm{\ridm{\kappa_1^\dagger \dec{qf}}\,\ridm{\kappa_1^\dagger
    \dec{pf}}} \\
    & = \ridm{\kappa_1^\dagger \dec{qf}}\,\ridm{\kappa_1^\dagger
    \dec{pf}}
  \end{align*}
  so $\kappa_1^\dagger (f^\diamond(\dec{p}) \land f^\diamond(\dec{q})) = 
  \kappa_1^\dagger (f^\diamond(\dec{p}) \land f^\diamond(\dec{q}))$. For the 
  second part,
  \begin{align*}
    \kappa_2^\dagger f^\diamond(\dec{p} \land \dec{q})
    & = \kappa_2^\dagger \dec{(\dec{p} \land \dec{q})f}
     = \ridm{\kappa_2^\dagger (\dec{p} \land \dec{q})f} \\
    & = \ridm{\kappa_2^\dagger (\id + \dec{p}^\dagger) \alpha (\dec{q} + 
    \ridm{\dec{q}}) \dec{p} f}
     = \ridm{\dec{p}^\dagger \kappa_2^\dagger \alpha (\dec{q} + 
    \ridm{\dec{q}}) \dec{p} f} \\
    & = \ridm{\dec{p}^\dagger (\kappa_2^\dagger + \id) (\dec{q} + 
    \ridm{\dec{q}}) \dec{p} f}
     = \ridm{\dec{p}^\dagger ((\kappa_2^\dagger \dec{q}) + 
    \ridm{\dec{q}}) \dec{p} f} \\
    & = \ridm{\ridm{\dec{p}^\dagger} (\ridm{\kappa_2^\dagger \dec{q}} + 
    \ridm{\dec{q}}) \dec{p} f}
     = \ridm{(\ridm{\kappa_2^\dagger \dec{q}} + 
    \ridm{\dec{q}}) \ridm{\dec{p}^\dagger} \dec{p} f} \\
    & = \ridm{(\ridm{\kappa_2^\dagger \dec{q}} + 
    \ridm{\dec{q}}) \dec{p} f}
     = \ridm{((\kappa_2^\dagger \dec{q}) + \dec{q}) \dec{p} f} \\
    & = \ridm{((\kappa_2^\dagger \dec{q}) + \dec{q}) \dec{p}\ridm{\dec{p}} f}
     = \ridm{((\kappa_2^\dagger \dec{q}\ridm{\dec{p}}) + 
    (\dec{q}\ridm{\dec{p}})) \dec{p} f} \\
    & = \ridm{((\kappa_2^\dagger \dec{q}\ridm{\dec{p}}f) + 
    (\dec{q}\ridm{\dec{p}}f)) \dec{pf}} \\
  \end{align*}
  and
  \begin{align*}
    \kappa_2^\dagger (f^\diamond(\dec{p}) \land f^\diamond(\dec{q}))
    & = \kappa_2^\dagger (\dec{pf} \land \dec{qf})
    = \ridm{\kappa_2^\dagger (\dec{pf} \land \dec{qf})} \\
    & = \ridm{\kappa_2^\dagger (\id + \dec{pf}^\dagger) \alpha (\dec{qf} + 
    \ridm{\dec{qf}})\dec{pf}}
    = \ridm{\dec{pf}^\dagger \kappa_2^\dagger \alpha (\dec{qf} + 
    \ridm{\dec{qf}})\dec{pf}} \\
    & = \ridm{\dec{pf}^\dagger (\kappa_2^\dagger + \id) (\dec{qf} + 
    \ridm{\dec{qf}})\dec{pf}}
    = \ridm{\dec{pf}^\dagger ((\kappa_2^\dagger\dec{qf}) + 
    \ridm{\dec{qf}})\dec{pf}} \\
    & = \ridm{\ridm{\dec{pf}^\dagger} (\ridm{\kappa_2^\dagger\dec{qf}} + 
    \ridm{\dec{qf}})\dec{pf}}
    = \ridm{(\ridm{\kappa_2^\dagger\dec{qf}} + 
    \ridm{\dec{qf}})\ridm{\dec{pf}^\dagger} \dec{pf}} \\
    & = \ridm{(\ridm{\kappa_2^\dagger\dec{qf}} + 
    \ridm{\dec{qf}}) \dec{pf}}
    = \ridm{(\ridm{\kappa_2^\dagger\dec{q}f} + 
    \ridm{\dec{q}f}) \dec{pf}} \\
    & = \ridm{((\kappa_2^\dagger\dec{q}f) + 
    (\dec{q}f)) \dec{pf}}
    = \ridm{((\kappa_2^\dagger\dec{q}f) + 
    (\dec{q}f)) \dec{pf} \ridm{\dec{pf}}} \\
    & = \ridm{((\kappa_2^\dagger\dec{q}f) + 
    (\dec{q}f)) \dec{pf} \ridm{\dec{p}f}}
    = \ridm{((\kappa_2^\dagger\dec{q}f\ridm{\dec{p}f}) + 
    (\dec{q}f\ridm{\dec{p}f})) \dec{pf}} \\
    & = \ridm{((\kappa_2^\dagger\dec{q}\ridm{\dec{p}}f) + 
    (\dec{q}\ridm{\dec{p}}f)) \dec{pf}}
  \end{align*}
  so also $\kappa_2^\dagger f^\diamond(\dec{p} \land \dec{q}) = 
  \kappa_2^\dagger (f^\diamond(\dec{p}) \land f^\diamond(\dec{q}))$, which 
  finally gives us $f^\diamond(\dec{p} \land \dec{q}) = f^\diamond(\dec{p}) 
  \land f^\diamond(\dec{q})$.
\end{proof}

\begin{proof}[Proof of \ref{lem:ridmdec}]
  \eqref{lem:ridmdec:1} $\ridm{\neg\dec{p}} = \ridm{\gamma \dec{p}} =
  \ridm{\ridm{\gamma} \dec{p}} = \ridm{\id \dec{p}} = \ridm{\dec{p}}$.
  
  \eqref{lem:ridmdec:2} We have
  \begin{align*}
    \ridm{\dec{p} \land \dec{q}} 
    & = \ridm{(\id + \dec{p}^\dagger) \alpha (\dec{q} + \ridm{\dec{q}}) 
    \dec{p}} \\
    & = \ridm{(\id + \nabla) \alpha (\dec{q} + \ridm{\dec{q}}) 
    \dec{p}} \\
    & = \ridm{\ridm{(\id + \nabla) \alpha} (\dec{q} + \ridm{\dec{q}}) 
    \dec{p}} \\
    & = \ridm{\id (\dec{q} + \ridm{\dec{q}}) 
    \dec{p}} \\
    & = \ridm{(\dec{q} + \ridm{\dec{q}}) \dec{p}} \\
    & = \ridm{(\ridm{\dec{q}} + \ridm{\dec{q}}) \dec{p}} \\
    & = \ridm{\dec{p} \ridm{\dec{q}}} \\
    & = \ridm{\dec{p}}\, \ridm{\dec{q}}
  \end{align*}
  \eqref{lem:ridmdec:3} follows by analogous reasoning to
  \eqref{lem:ridmdec:2}. We have \eqref{lem:ridmdec:4} immediately by
  \eqref{lem:ridmdec:2} since $\ridm{\dec{p} \land \dec{q}} =
  \ridm{\dec{p}}\,\ridm{\dec{q}} \le \ridm{\dec{p}}$ directly, and likewise
  $\ridm{\dec{p} \land \dec{q}} = \ridm{\dec{p}}\,\ridm{\dec{q}} \le
  \ridm{\dec{q}}$. \eqref{lem:ridmdec:5} follows analogously.
\end{proof}

\begin{proof}[Proof of Lemma~\ref{lem:entailment}]
  Assume $\kappa_1^\dagger \dec{p} \le \kappa_1^\dagger \dec{q}$ and
  $\ridm{\dec{p}} \le \ridm{\dec{q}}$. By Lemma~\ref{lem:decrep}, to show
  $\dec{p} \land \dec{q} = \dec{p}$ (\ie, $\dec{p} \entails \dec{q}$) it
  suffices to show that $\kappa_1^\dagger (\dec{p} \land \dec{q}) =
  \kappa_1^\dagger \dec{p}$ and $\kappa_2^\dagger (\dec{p} \land \dec{q}) =
  \kappa_2^\dagger \dec{p}$.
  
  Since $$\kappa_1^\dagger (\dec{p} \land \dec{q}) = \ridm{\kappa_1^\dagger
  p}\,\ridm{\kappa_1^\dagger q} = \ridm{\kappa_1^\dagger
  \dec{p}}\,\ridm{\kappa_1^\dagger \dec{q}} = \ridm{\kappa_1^\dagger \dec{p}} =
  \kappa_1^\dagger \dec{p}$$ by the proof of Lemma~\ref{lem:condec},
  $\ridm{\kappa_1^\dagger \dec{p}} \le \ridm{\kappa_1^\dagger \dec{q}}$, and 
  Lemma~\ref{lem:utility}, proving the first part. For the second part,
  \begin{align*}
    \kappa_2^\dagger (\dec{p} \land \dec{q})
    & = \ridm{\kappa_2^\dagger (\dec{p} \land \dec{q})} 
     = \ridm{\dec{p}^\dagger (\ridm{\kappa_2^\dagger \dec{q}} + 
    \ridm{\dec{q}}) \dec{p}} \\
    & = \ridm{\ridm{\dec{p}^\dagger} (\ridm{\kappa_2^\dagger \dec{q}} + 
    \ridm{\dec{q}}) \dec{p}}
     = \ridm{(\ridm{\kappa_1^\dagger \dec{p}}+\ridm{\kappa_2^\dagger \dec{p}})
    (\ridm{\kappa_2^\dagger \dec{q}} + \ridm{\dec{q}}) \dec{p}} \\
    & = \ridm{((\ridm{\kappa_1^\dagger \dec{p}}\,\ridm{\kappa_2^\dagger
    \dec{q}}) + (\ridm{\kappa_2^\dagger \dec{p}}\,\ridm{\dec{q}})) \dec{p}}
     = \ridm{((\ridm{\kappa_1^\dagger \dec{p}}\,\ridm{\kappa_1^\dagger
    \dec{q}}\,\ridm{\kappa_2^\dagger \dec{q}}) + (\ridm{\kappa_2^\dagger
    \dec{p}}\,\ridm{\dec{q}})) \dec{p}} \\
    & = \ridm{((\ridm{\kappa_1^\dagger \dec{p}}\,0) + (\ridm{\kappa_2^\dagger
    \dec{p}}\,\ridm{\dec{q}})) \dec{p}}
     = \ridm{(0 + (\ridm{\kappa_2^\dagger \dec{p}}\,\ridm{\dec{q}})) \dec{p}}\\
    & = \ridm{(0 + (\ridm{\kappa_2^\dagger \dec{p}\ridm{\dec{q}}})) \dec{p}}
     = \ridm{(0 + (\ridm{\kappa_2^\dagger \dec{p} \ridm{\dec{p}}\,
    \ridm{\dec{q}}})) \dec{p}} \\
    & = \ridm{(0 + (\ridm{\kappa_2^\dagger \dec{p} \ridm{\dec{p}}})) \dec{p}}
     = \ridm{(0 + (\ridm{\kappa_2^\dagger \dec{p}})) \dec{p}} \\
    & = \ridm{((0\,\ridm{\kappa_1^\dagger \dec{p}}) + (\ridm{\kappa_2^\dagger 
    \dec{p}})) \dec{p}}
     = \ridm{(0 + \id) (\ridm{\kappa_1^\dagger \dec{p}} +
    \ridm{\kappa_2^\dagger \dec{p}}) \dec{p}} \\
    & = \ridm{(0 + \id) \ridm{\dec{p}^\dagger} \dec{p}}
     = \ridm{(0 + \id) \dec{p}}
     = \ridm{\ridm{\kappa_2^\dagger} \dec{p}}
     = \ridm{\kappa_2^\dagger \dec{p}}
     = \kappa_2^\dagger \dec{p}
  \end{align*}
  In the other direction, suppose that $\dec{p} \entails \dec{q}$, \ie,
  $\dec{p} \land \dec{q} = \dec{p}$. Then $\kappa_1^\dagger (\dec{p} \land
  \dec{q}) = \kappa_1^\dagger \dec{p} = \ridm{\kappa_1^\dagger \dec{p}}$, but
  since we also know that $\kappa_1^\dagger (\dec{p} \land \dec{q}) =
  \ridm{\kappa_1^\dagger \dec{p}}\,\ridm{\kappa_1^\dagger \dec{q}}$ (see
  above), it follows that $$\kappa_1^\dagger \dec{q}\ridm{\kappa_1^\dagger
  \dec{p}} = \ridm{\kappa_1^\dagger \dec{q}}\,\ridm{\kappa_1^\dagger \dec{p}} =
  \ridm{\kappa_1^\dagger \dec{p}}\,\ridm{\kappa_1^\dagger \dec{q}} =
  \ridm{\kappa_1^\dagger \dec{p}} = \kappa_1^\dagger \dec{p}$$ that is,
  $\kappa_1^\dagger \dec{p} \le \kappa_1^\dagger \dec{q}$. That $\ridm{\dec{p}}
  \le \ridm{\dec{q}}$ follows by Lemma~\ref{lem:ridmdec}, as we thus have
  $\ridm{\dec{p}} = \ridm{\dec{p} \land \dec{q}} \le \ridm{\dec{q}}$.
\end{proof}

\begin{proof}[Proof of Lemma~\ref{lem:dec_misc}]
  To prove \eqref{lem:commstmt}, it suffices to show that their partial
  inverses are equal, since partial inverses are unique. We show this as 
  follows:
  \ctikzfig{commutativity}
  \eqref{lem:con_irrev} follows by
  \ctikzfig{expression_lemma2}
  \eqref{lem:dis_irrev} follows by analogous
  argument to \eqref{lem:con_irrev}.
\end{proof}

\begin{proof}[Proof of Lemma~\ref{lem:dec_dmq}]
\noindent
Commutativity of conjunction follows by
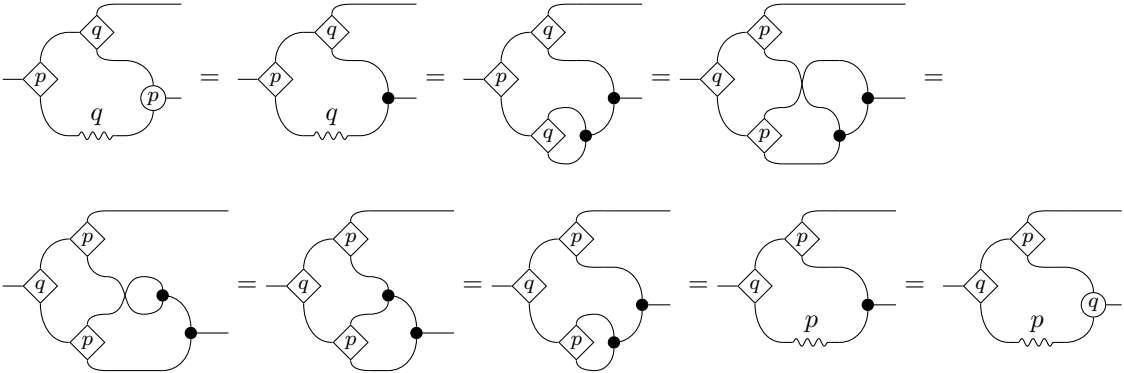
\ctikzfig{con_comm}
and commutativity of disjunction analogously. Associativity of conjunction is demonstrated by
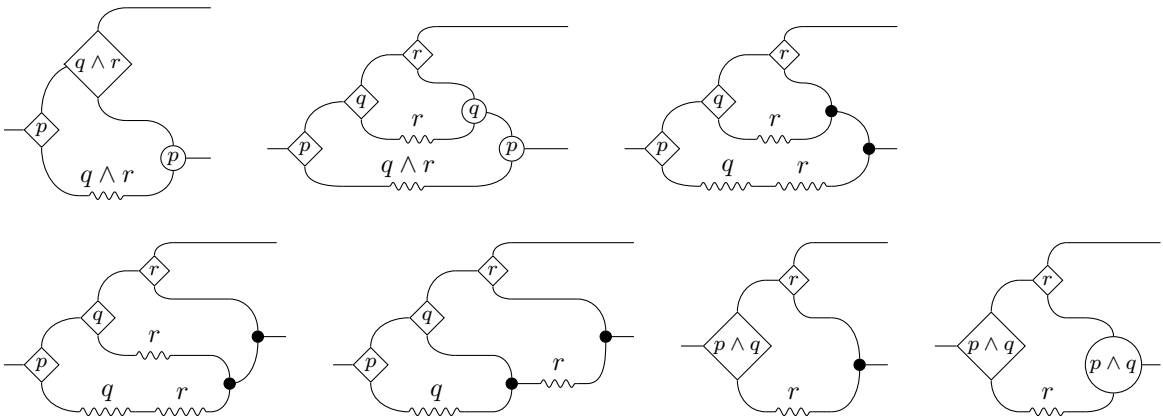
\ctikzfig{con_assoc}
and associativity of disjunction can be shown similarly. For distributivity of conjunction over disjunction,
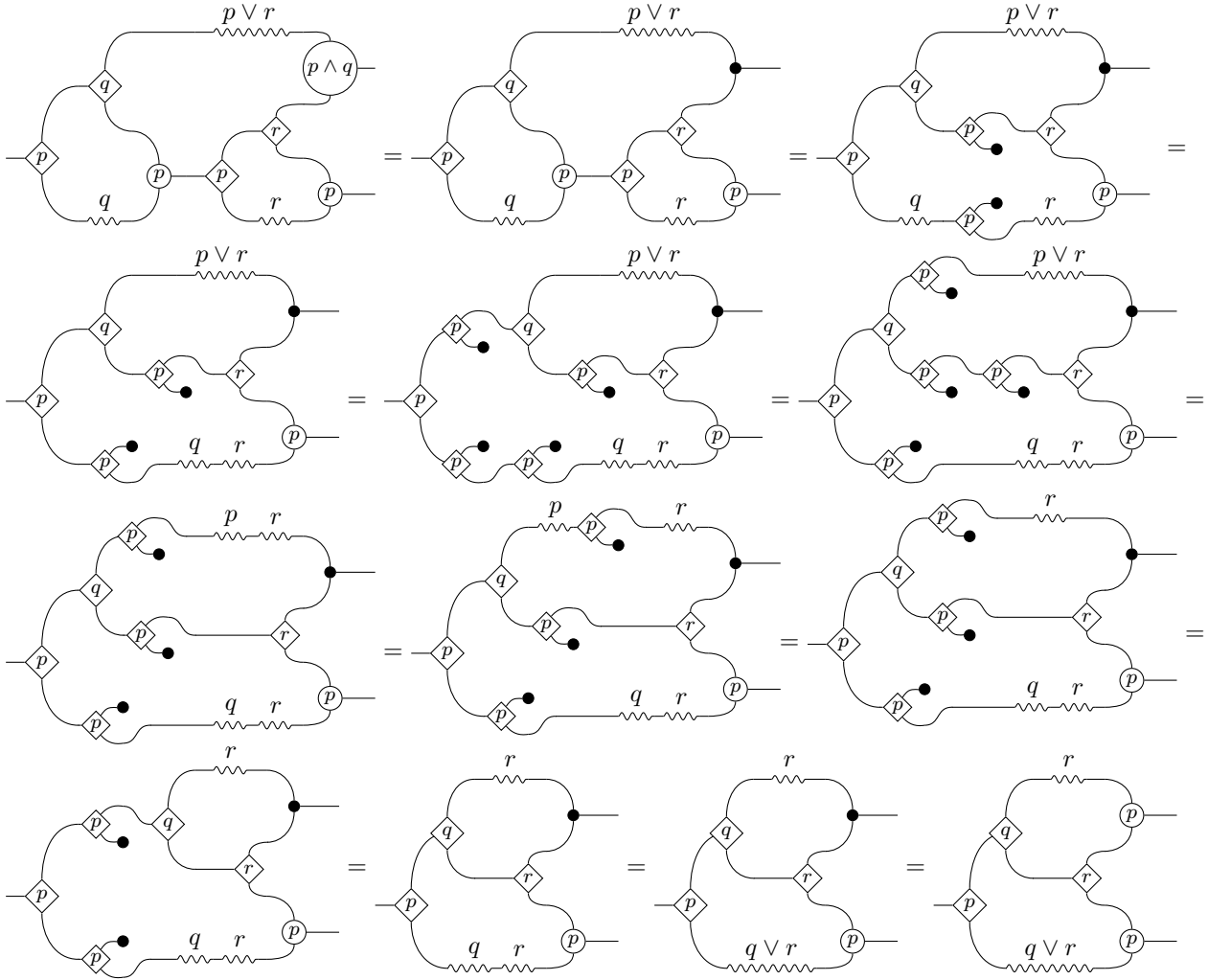
\ctikzfig{con_dist}
and the dual distributive law follows symmetrically. Finally, the double negation law then follows simply by
\ctikzfig{dne}
which concludes the proof.
\end{proof}
\fi


\end{document}

%% file: stringdiagrams.tikzstyles
\tikzstyle{small box}=[fill=white, draw=black, shape=rectangle, tikzit fill=white, tikzit draw=black, font={\small}]
\tikzstyle{large box}=[fill=white, draw=black, shape=rectangle, tikzit fill=white, tikzit draw=black, minimum height=15mm, minimum width=5mm]
\tikzstyle{dot}=[fill=black, shape=circle, minimum size=1.5mm, inner sep=0pt, tikzit shape=circle, tikzit fill=black, tikzit draw=black, draw=black]
\tikzstyle{pbox}=[fill=white, draw=black, shape=diamond, tikzit fill=white, tikzit draw=black, font={\small}, inner sep=1pt, tikzit shape=rectangle]
\tikzstyle{tiny box}=[fill=white, draw=black, shape=rectangle, tikzit fill=white, tikzit draw=black, font={\footnotesize}, inner sep=2pt]
\tikzstyle{whitedot}=[fill=white, draw=black, shape=circle, tikzit draw=black, tikzit fill=white, tikzit shape=circle, inner sep=0pt, minimum size=1.5mm]
\tikzstyle{pcirc}=[fill=white, draw=black, shape=circle, tikzit fill=white, tikzit draw=black, tikzit shape=circle, inner sep=1pt, font={\small}]
\tikzstyle{tiny}=[font={\footnotesize}]
\tikzstyle{pboxtiny}=[fill=white, draw=black, shape=diamond, tikzit fill=white, tikzit draw=black, tikzit shape=rectangle, inner sep=0.5pt, font={\scriptsize}]
\tikzstyle{actually_tiny}=[font={\tiny}]

\tikzstyle{ridm}=[-, decorate, decoration={{snake, segment length=1.5mm, amplitude=0.5mm}}]
\tikzstyle{arrow}=[->]

%% file: pdec.tikz
\begin{tikzpicture}
	\begin{pgfonlayer}{nodelayer}
		\node [style=none] (0) at (0, 0.75) {};
		\node [style=none] (1) at (0, 0.75) {};
		\node [style=none] (2) at (1.5, 1.5) {};
		\node [style=pbox] (3) at (0.75, 0.75) {$p$};
		\node [style=none] (4) at (1.5, 0) {};
	\end{pgfonlayer}
	\begin{pgfonlayer}{edgelayer}
		\draw (1.center) to (3);
		\draw [in=180, out=90] (3) to (2.center);
		\draw [in=-180, out=-90] (3) to (4.center);
	\end{pgfonlayer}
\end{tikzpicture}

%% file: qdec.tikz
\begin{tikzpicture}
	\begin{pgfonlayer}{nodelayer}
		\node [style=none] (0) at (0, 0.75) {};
		\node [style=none] (1) at (0, 0.75) {};
		\node [style=none] (2) at (1.5, 1.5) {};
		\node [style=pbox] (3) at (0.75, 0.75) {$q$};
		\node [style=none] (4) at (1.5, 0) {};
	\end{pgfonlayer}
	\begin{pgfonlayer}{edgelayer}
		\draw (1.center) to (3);
		\draw [in=180, out=90] (3) to (2.center);
		\draw [in=-180, out=-90] (3) to (4.center);
	\end{pgfonlayer}
\end{tikzpicture}

%% file: commstmt.tikz
\begin{tikzpicture}
	\begin{pgfonlayer}{nodelayer}
		\node [style=none] (0) at (0.25, 2.75) {};
		\node [style=none] (1) at (2.25, 4.25) {};
		\node [style=none] (2) at (2.25, 1.25) {};
		\node [style=pbox] (3) at (1.5, 2.75) {$p$};
		\node [style=none] (4) at (2.25, 4.25) {};
		\node [style=none] (5) at (3.25, 5.25) {};
		\node [style=none] (6) at (3.25, 3.25) {};
		\node [style=pbox] (7) at (2.5, 4.25) {$q$};
		\node [style=none] (8) at (2.25, 1.25) {};
		\node [style=none] (9) at (3.25, 2.25) {};
		\node [style=none] (10) at (3.25, 0.25) {};
		\node [style=pbox] (11) at (2.5, 1.25) {$q$};
		\node [style=none] (12) at (3.25, 3.25) {};
		\node [style=none] (13) at (3.25, 2.25) {};
		\node [style=none] (14) at (4.25, 3.25) {};
		\node [style=none] (15) at (4.25, 2.25) {};
		\node [style=none] (16) at (4.25, 0.25) {};
		\node [style=none] (17) at (4.25, 5.25) {};
		\node [style=none] (18) at (6.75, 2.75) {};
		\node [style=none] (19) at (8.75, 4.25) {};
		\node [style=none] (20) at (8.75, 1.25) {};
		\node [style=pbox] (21) at (8, 2.75) {$q$};
		\node [style=none] (22) at (8.75, 4.25) {};
		\node [style=none] (23) at (9.75, 5.25) {};
		\node [style=none] (24) at (9.75, 3.25) {};
		\node [style=pbox] (25) at (9, 4.25) {$p$};
		\node [style=none] (26) at (8.75, 1.25) {};
		\node [style=none] (27) at (9.75, 2.25) {};
		\node [style=none] (28) at (9.75, 0.25) {};
		\node [style=pbox] (29) at (9, 1.25) {$p$};
		\node [style=none] (30) at (9.75, 3.25) {};
		\node [style=none] (31) at (9.75, 2.25) {};
		\node [style=none] (32) at (5.5, 2.75) {$=$};
	\end{pgfonlayer}
	\begin{pgfonlayer}{edgelayer}
		\draw (0.center) to (3);
		\draw [in=180, out=90] (3) to (1.center);
		\draw [in=180, out=-90] (3) to (2.center);
		\draw (4.center) to (7);
		\draw [in=180, out=90] (7) to (5.center);
		\draw [in=180, out=-90] (7) to (6.center);
		\draw (8.center) to (11);
		\draw [in=180, out=90] (11) to (9.center);
		\draw [in=180, out=-90] (11) to (10.center);
		\draw [in=-180, out=0, looseness=1.25] (12.center) to (15.center);
		\draw [in=-180, out=0, looseness=1.25] (13.center) to (14.center);
		\draw (5.center) to (17.center);
		\draw (10.center) to (16.center);
		\draw (18.center) to (21);
		\draw [in=180, out=90] (21) to (19.center);
		\draw [in=180, out=-90] (21) to (20.center);
		\draw (22.center) to (25);
		\draw [in=180, out=90] (25) to (23.center);
		\draw [in=180, out=-90] (25) to (24.center);
		\draw (26.center) to (29);
		\draw [in=180, out=90] (29) to (27.center);
		\draw [in=180, out=-90] (29) to (28.center);
	\end{pgfonlayer}
\end{tikzpicture}

%% file: con_irrev.tikz
\begin{tikzpicture}
	\begin{pgfonlayer}{nodelayer}
		\node [style=none] (70) at (0, 1.75) {};
		\node [style=none] (71) at (1.75, 3) {};
		\node [style=none] (72) at (1.75, 0.25) {};
		\node [style=pbox] (73) at (1, 1.75) {$p$};
		\node [style=none] (74) at (1.75, 3) {};
		\node [style=none] (75) at (3.25, 3.75) {};
		\node [style=pbox] (76) at (2.5, 3) {$q$};
		\node [style=none] (77) at (1.75, 0.25) {};
		\node [style=none] (78) at (4.75, 3.75) {};
		\node [style=none] (79) at (2, 0.25) {};
		\node [style=none] (80) at (3, 0.25) {};
		\node [style=none] (81) at (1.75, 0.25) {};
		\node [style=none] (82) at (2.5, 0.75) {$q$};
		\node [style=pcirc] (83) at (4, 1.25) {$p$};
		\node [style=none] (84) at (3.25, 2.25) {};
		\node [style=none] (85) at (3.25, 0.25) {};
		\node [style=none] (86) at (4.75, 1.25) {};
		\node [style=none] (87) at (7.25, 1.75) {};
		\node [style=none] (88) at (9, 3) {};
		\node [style=none] (89) at (9, 0.25) {};
		\node [style=pbox] (90) at (8.25, 1.75) {$p$};
		\node [style=none] (91) at (9, 3) {};
		\node [style=none] (92) at (10.5, 3.75) {};
		\node [style=pbox] (93) at (9.75, 3) {$q$};
		\node [style=none] (94) at (9, 0.25) {};
		\node [style=none] (95) at (12, 3.75) {};
		\node [style=none] (96) at (9.25, 0.25) {};
		\node [style=none] (97) at (10.25, 0.25) {};
		\node [style=none] (98) at (9, 0.25) {};
		\node [style=none] (99) at (9.75, 0.75) {$q$};
		\node [style=dot] (100) at (11.25, 1.25) {};
		\node [style=none] (101) at (10.5, 2.25) {};
		\node [style=none] (102) at (10.5, 0.25) {};
		\node [style=none] (103) at (12, 1.25) {};
		\node [style=none] (104) at (6, 1.75) {$=$};
	\end{pgfonlayer}
	\begin{pgfonlayer}{edgelayer}
		\draw (70.center) to (73);
		\draw [in=180, out=90] (73) to (71.center);
		\draw [in=180, out=-90] (73) to (72.center);
		\draw (74.center) to (76);
		\draw [in=180, out=90] (76) to (75.center);
		\draw (75.center) to (78.center);
		\draw [style=ridm] (79.center) to (80.center);
		\draw (81.center) to (79.center);
		\draw [in=-90, out=0] (85.center) to (83);
		\draw [in=90, out=0] (84.center) to (83);
		\draw (83) to (86.center);
		\draw [in=-180, out=-90] (76) to (84.center);
		\draw (80.center) to (85.center);
		\draw (87.center) to (90);
		\draw [in=180, out=90] (90) to (88.center);
		\draw [in=180, out=-90] (90) to (89.center);
		\draw (91.center) to (93);
		\draw [in=180, out=90] (93) to (92.center);
		\draw (92.center) to (95.center);
		\draw [style=ridm] (96.center) to (97.center);
		\draw (98.center) to (96.center);
		\draw [in=-90, out=0] (102.center) to (100);
		\draw [in=90, out=0] (101.center) to (100);
		\draw (100) to (103.center);
		\draw [in=-180, out=-90] (93) to (101.center);
		\draw (97.center) to (102.center);
	\end{pgfonlayer}
\end{tikzpicture}

%% file: dis_irrev.tikz
\begin{tikzpicture}
	\begin{pgfonlayer}{nodelayer}
		\node [style=none] (0) at (0, 2.25) {};
		\node [style=none] (1) at (1.75, 3.75) {};
		\node [style=none] (2) at (1.75, 1) {};
		\node [style=pbox] (3) at (1, 2.25) {$p$};
		\node [style=none] (4) at (1.75, 3.75) {};
		\node [style=none] (5) at (1.75, 3.75) {};
		\node [style=none] (6) at (3.25, 3.75) {};
		\node [style=none] (7) at (3.25, 3.75) {};
		\node [style=none] (8) at (2, 3.75) {};
		\node [style=none] (9) at (3, 3.75) {};
		\node [style=none] (10) at (3.25, 3.75) {};
		\node [style=none] (11) at (1.75, 3.75) {};
		\node [style=none] (12) at (1.75, 1) {};
		\node [style=none] (13) at (1.75, 1) {};
		\node [style=none] (14) at (3.25, 1.75) {};
		\node [style=none] (15) at (3.25, 0.25) {};
		\node [style=pbox] (16) at (2.5, 1) {$q$};
		\node [style=none] (17) at (3.25, 0.25) {};
		\node [style=none] (18) at (3.25, 3.75) {};
		\node [style=none] (19) at (3.25, 1.75) {};
		\node [style=none] (20) at (3.25, 3.75) {};
		\node [style=none] (21) at (3.25, 1.75) {};
		\node [style=none] (22) at (3.25, 1.75) {};
		\node [style=none] (23) at (4.75, 0.25) {};
		\node [style=pcirc] (24) at (4, 2.75) {$p$};
		\node [style=none] (25) at (3.25, 3.75) {};
		\node [style=none] (26) at (3.25, 1.75) {};
		\node [style=none] (27) at (4.75, 2.75) {};
		\node [style=none] (28) at (7.25, 2.25) {};
		\node [style=none] (29) at (9, 3.75) {};
		\node [style=none] (30) at (9, 1) {};
		\node [style=pbox] (31) at (8.25, 2.25) {$p$};
		\node [style=none] (32) at (9, 3.75) {};
		\node [style=none] (33) at (9, 3.75) {};
		\node [style=none] (34) at (10.5, 3.75) {};
		\node [style=none] (35) at (10.5, 3.75) {};
		\node [style=none] (36) at (9.25, 3.75) {};
		\node [style=none] (37) at (10.25, 3.75) {};
		\node [style=none] (38) at (10.5, 3.75) {};
		\node [style=none] (39) at (9, 3.75) {};
		\node [style=none] (40) at (9, 1) {};
		\node [style=none] (41) at (9, 1) {};
		\node [style=none] (42) at (10.5, 1.75) {};
		\node [style=none] (43) at (10.5, 0.25) {};
		\node [style=pbox] (44) at (9.75, 1) {$q$};
		\node [style=none] (45) at (10.5, 0.25) {};
		\node [style=none] (46) at (10.5, 3.75) {};
		\node [style=none] (47) at (10.5, 1.75) {};
		\node [style=none] (48) at (10.5, 3.75) {};
		\node [style=none] (49) at (10.5, 1.75) {};
		\node [style=none] (50) at (10.5, 1.75) {};
		\node [style=none] (51) at (12, 0.25) {};
		\node [style=dot] (52) at (11.25, 2.75) {};
		\node [style=none] (53) at (10.5, 3.75) {};
		\node [style=none] (54) at (10.5, 1.75) {};
		\node [style=none] (55) at (12, 2.75) {};
		\node [style=none] (56) at (6, 2.25) {$=$};
		\node [style=none] (57) at (2.5, 4.25) {$q$};
		\node [style=none] (58) at (9.75, 4.25) {$q$};
	\end{pgfonlayer}
	\begin{pgfonlayer}{edgelayer}
		\draw (0.center) to (3);
		\draw [in=180, out=90] (3) to (1.center);
		\draw [in=180, out=-90] (3) to (2.center);
		\draw [style=ridm] (8.center) to (9.center);
		\draw (11.center) to (8.center);
		\draw (10.center) to (9.center);
		\draw (13.center) to (16);
		\draw [in=180, out=90] (16) to (14.center);
		\draw [in=180, out=-90] (16) to (15.center);
		\draw (17.center) to (23.center);
		\draw [in=-90, out=0] (26.center) to (24);
		\draw [in=90, out=0] (25.center) to (24);
		\draw (24) to (27.center);
		\draw (28.center) to (31);
		\draw [in=180, out=90] (31) to (29.center);
		\draw [in=180, out=-90] (31) to (30.center);
		\draw [style=ridm] (36.center) to (37.center);
		\draw (39.center) to (36.center);
		\draw (38.center) to (37.center);
		\draw (41.center) to (44);
		\draw [in=180, out=90] (44) to (42.center);
		\draw [in=180, out=-90] (44) to (43.center);
		\draw (45.center) to (51.center);
		\draw [in=-90, out=0] (54.center) to (52);
		\draw [in=90, out=0] (53.center) to (52);
		\draw (52) to (55.center);
	\end{pgfonlayer}
\end{tikzpicture}

%% file: ridmfg.tikz
\begin{tikzpicture}
	\begin{pgfonlayer}{nodelayer}
		\node [style=none] (0) at (0.25, 0) {};
		\node [style=none] (1) at (1.75, 0) {};
		\node [style=tiny] (2) at (1, 0.5) {\tiny$gf$};
		\node [style=none] (3) at (2, 0) {};
		\node [style=none] (4) at (0, 0) {};
	\end{pgfonlayer}
	\begin{pgfonlayer}{edgelayer}
		\draw [style=ridm] (0.center) to (1.center);
		\draw (4.center) to (0.center);
		\draw (3.center) to (1.center);
	\end{pgfonlayer}
\end{tikzpicture}

%% file: mfps.bbl
\begin{thebibliography}{10}

\bibitem{Ammann2008}
P.~Ammann and J.~Offutt.
\newblock {\em Introduction to Software Testing}.
\newblock Cambridge University Press, 1st edition, 2008.

\bibitem{ManesArbib1980}
M.~A. Arbib and E.~G. Manes.
\newblock Partially additive categories and flow-diagram semantics.
\newblock {\em Journal of Algebra}, 62(1):203 -- 227, 1980.

\bibitem{AshcroftManna1972}
E.~Ashcroft and Z.~Manna.
\newblock The translation of 'go to' programs into 'while' programs.
\newblock In C.~V. Freiman, J.~E. Griffith, and J.~L. Rosenfeld, editors, {\em
  Proceedings of IFIP Congress 71}, volume~1, pages 250--255. North-Holland,
  1972.

\bibitem{BohmJacopini1966}
C.~B\"{o}hm and G.~Jacopini.
\newblock Flow diagrams, {T}uring machines and languages with only two
  formation rules.
\newblock {\em Communications of the ACM}, 9(5):366--371, 1966.

\bibitem{Carboni1993}
A.~Carboni, S.~Lack, and R.~Walters.
\newblock Introduction to extensive and distributive categories.
\newblock {\em Journal of Pure and Applied Algebra}, 84(2):145--158, 1993.

\bibitem{CJWW2015}
K.~Cho, B.~Jacobs, B.~Westerbaan, and A.~Westerbaan.
\newblock An introduction to effectus theory.
\newblock See \url{http://arxiv.org/abs/1512.05813}, 2015.

\bibitem{Cockett2002}
J.~R.~B. Cockett and S.~Lack.
\newblock Restriction categories {I}: {C}ategories of partial maps.
\newblock {\em Theoretical Computer Science}, 270(1--2):223--259, 2002.

\bibitem{Cockett2003}
J.~R.~B. Cockett and S.~Lack.
\newblock Restriction categories {II}: {P}artial map classification.
\newblock {\em Theoretical Computer Science}, 294(1--2):61--102, 2003.

\bibitem{Cockett2007}
R.~Cockett and S.~Lack.
\newblock Restriction categories {III}: {C}olimits, partial limits and
  extensivity.
\newblock {\em Mathematical Structures in Computer Science}, 17(4):775--817,
  2007.

\bibitem{Elgot1975}
C.~C. Elgot.
\newblock Monadic computation and iterative algebraic theories.
\newblock In H.~E. Rose and J.~C. Shepherdson, editors, {\em Logic Colloquium},
  pages 175--230. North Holland, 1975.

\bibitem{Elgot1976}
C.~C. Elgot.
\newblock Structured programming with and without {GO TO} statements.
\newblock {\em IEEE Transactions on Software Engineering}, SE-2:41--53, 1976.

\bibitem{FinnGrigolia1993}
V.~K. Finn and R.~Grigolia.
\newblock Nonsense logics and their algebraic properties.
\newblock {\em Theoria}, 59(1--3):207--273, 1993.

\bibitem{Glueck2017}
R.~Gl\"{u}ck and R.~Kaarsgaard.
\newblock A categorical foundation for structured reversible flowchart
  languages.
\newblock In A.~Silva, editor, {\em The Thirty-third Conference on the
  Mathematical Foundations of Programming Semantics (MFPS XXXIII)}, volume 336
  of {\em Electronic Notes in Theoretical Computer Science}. Elsevier, 2017.

\bibitem{Guo2012}
X.~Guo.
\newblock {\em {Products, Joins, Meets, and Ranges in Restriction Categories}}.
\newblock PhD thesis, University of Calgary, 2012.

\bibitem{Jacobs1999}
B.~Jacobs.
\newblock {\em Categorical Logic and Type Theory}, volume 141 of {\em Studies
  in Logic and the Foundations of Mathematics}.
\newblock Elsevier, first edition, 1999.

\bibitem{Jacobs2015}
B.~Jacobs.
\newblock New directions in categorical logic, for classical, probabilistic and
  quantum logic.
\newblock {\em Logical Methods in Computer Science}, 11(3), 2015.

\bibitem{JonesGomardSestoft1993}
N.~D. Jones, C.~K. Gomard, and P.~Sestoft.
\newblock {\em Partial Evaluation and Automatic Program Generation}.
\newblock Prentice Hall International, 1993.

\bibitem{Kaarsgaard2017}
R.~Kaarsgaard, H.~B. Axelsen, and R.~Glück.
\newblock Join inverse categories and reversible recursion.
\newblock {\em Journal of Logical and Algebraic Methods in Programming},
  87:33--50, 2017.

\bibitem{KaarsgaardGlueck2018}
R.~Kaarsgaard and R.~Gl\"{u}ck.
\newblock A categorical foundation for structured reversible flowchart
  languages: {S}oundness and adequacy.
\newblock {\em Logical Methods in Computer Science}, 14(3):1--38, 2018.

\bibitem{Kastl1979}
J.~Kastl.
\newblock Inverse categories.
\newblock In H.-J. Hoehnke, editor, {\em Algebraische Modelle, Kategorien und
  Gruppoide}, volume~7 of {\em Studien zur Algebra und ihre Anwendungen}, pages
  51--60. Akademie-Verlag, 1979.

\bibitem{Kleene1952}
S.~C. Kleene.
\newblock {\em Introduction to metamathematics}.
\newblock North Holland, 1st edition, 1952.

\bibitem{Lawvere1969}
F.~W. Lawvere.
\newblock Adjointness in foundations.
\newblock {\em Dialectica}, 23:281--296, 1969.

\bibitem{Ledda2018}
A.~Ledda.
\newblock Stone-type representations and dualities for varieties of
  bisemilattices.
\newblock {\em Studia Logica}, 106(2):417--448, 2018.

\bibitem{ManesArbib1986}
E.~G. Manes and M.~A. Arbib.
\newblock {\em {Algebraic approaches to program semantics}}.
\newblock Springer, 1986.

\bibitem{Selinger2011}
P.~Selinger.
\newblock A survey of graphical languages for monoidal categories.
\newblock In B.~Coecke, editor, {\em New Structures for Physics}, pages
  289--355. Springer, 2011.

\bibitem{Stefanescu1986}
G.~{\c{S}}tef{\u{a}}nescu.
\newblock An algebraic theory of flowchart schemes.
\newblock In P.~Franchi-Zannettacci, editor, {\em CAAP '86}, volume 214 of {\em
  Lecture Notes in Computer Science}, pages 60--73. Springer, 1986.

\bibitem{Stefanescu1987}
G.~{\c{S}}tef{\u{a}}nescu.
\newblock On flowchart theories part {I}: {T}he deterministic case.
\newblock {\em Journal of Computer and System Sciences}, 35(2):163 -- 191,
  1987.

\bibitem{WilliamsOssher1978}
M.~Williams and H.~Ossher.
\newblock Conversion of unstructured flow diagrams into structured form.
\newblock {\em The Computer Journal}, 21(2):161--167, 1978.

\bibitem{YokoyamaAG2016}
T.~Yokoyama, H.~B. Axelsen, and R.~Gl\"{u}ck.
\newblock Fundamentals of reversible flowchart languages.
\newblock {\em Theoretical Computer Science}, 611:87--115, 2016.

\end{thebibliography}
